%% file: main.tex
\documentclass[10pt, pra,notitlepage,twocolumn,superscriptaddress,longbibliography]{revtex4-2}

\input{pretex}

\usepackage{times}
\usepackage{amsmath,amsfonts,amssymb,graphicx,mathtools,bm,tcolorbox,relsize}
\usepackage[utf8]{inputenc}
\usepackage[T1]{fontenc}
\usepackage[qm]{qcircuit}
\pgfplotsset{compat=1.18}
\usepackage{braket}
\usepackage{natbib}
\usepackage{multirow}
\usepackage{booktabs}
\usepackage{array}

\usepackage{tabularx}
\usepackage{soul}
\usepackage[export]{adjustbox}

\definecolor{colortwo}{rgb}{0.4,0.77,0.17}
\definecolor{colorthree}{rgb}{0.01,0.51,0.93}
\definecolor{darkgray}{rgb}{0.3,0.3,0.3}

\newcommand{\update}[1]{#1}

\newcommand{\highlight}[1]{\textcolor{colorthree}{#1}}
\newcommand{\lowlight}[1]{\textcolor{darkgray}{#1}}

\nc{\Uin}{ U_\textrm{in} }
\nc{\Nin}{ \cN_\textrm{in} }
\nc{\Nout}{ \cN_\textrm{out} }
\nc{\su}{ \operatorname{SU} }
\newcommand{\aux}{ \text{aux} }
\newcommand{\bmP}{ \mathbf{P} }
\newcommand{\bmF}{ \mathbf{F} }
\newcommand{\bmV}{ \mathbf{V} }

\newcommand{\uni}[1]{ \textrm{Universal}_{#1} }
\newcommand{\fourcir}{ \mathfrak{C}_\mathrm{IV} }
\newcommand{\fivecir}{ \mathfrak{C}_\mathrm{V} }

\nc{\cmark}{\ding{51}}
\nc{\xmark}{\ding{55}}

\newcommand{\trace}[2][]{ \tr_{#1}\left[ #2 \right] }
\newcommand{\integral}[2][]{ \int_{#1} \text{d}{#2}\, }

\renewcommand{\set}[1]{ \left\{ #1 \right\} }

\nc{\kett}[1]{\lvert #1 \rangle\!\rangle}
\nc{\ketbravec}[2]{\lvert #1 \rangle\!\rangle\!\langle\!\langle #2 \rvert}

\SetKwIF{If}{ElseIf}{Else}{if}{}{else if}{else}{end if}%
\SetKwFor{While}{while}{}{end while}%
\SetKwRepeat{Do}{do}{while}
\SetKw{KwGoTo}{go to}

\allowdisplaybreaks

\begin{document}
\title{Parameterized quantum comb and simpler circuits for reversing unknown qubit-unitary operations}
\author{Yin Mo}
\thanks{Yin Mo, Lei Zhang, and Yu-Ao Chen contributed equally to this work.}
\affiliation{Thrust of Artificial Intelligence, Information Hub, The Hong Kong University of Science and Technology (Guangzhou), Guangdong, China}
\author{Lei Zhang}
\thanks{Yin Mo, Lei Zhang, and Yu-Ao Chen contributed equally to this work.}
\affiliation{Thrust of Artificial Intelligence, Information Hub, The Hong Kong University of Science and Technology (Guangzhou), Guangdong, China}
\author{Yu-Ao Chen}
\thanks{Yin Mo, Lei Zhang, and Yu-Ao Chen contributed equally to this work.}
\affiliation{Thrust of Artificial Intelligence, Information Hub, The Hong Kong University of Science and Technology (Guangzhou), Guangdong, China}
\author{Yingjian Liu}
\affiliation{Thrust of Artificial Intelligence, Information Hub, The Hong Kong University of Science and Technology (Guangzhou), Guangdong, China}
\author{Tengxiang Lin}
\affiliation{Thrust of Artificial Intelligence, Information Hub, The Hong Kong University of Science and Technology (Guangzhou), Guangdong, China}
\author{Xin Wang}
\email{felixxinwang@hkust-gz.edu.cn}
\affiliation{Thrust of Artificial Intelligence, Information Hub, The Hong Kong University of Science and Technology (Guangzhou), Guangdong, China}

\begin{abstract}
Quantum combs play a vital role in characterizing and transforming quantum processes, with wide-ranging applications in quantum information processing. However, obtaining the explicit quantum circuit for the desired quantum comb remains a challenging problem. We propose PQComb, a novel framework that employs parameterized quantum circuits (PQCs) or quantum neural networks to harness the full potential of quantum combs for diverse quantum process transformation tasks. This method is well-suited for near-term quantum devices and can be applied to various tasks in quantum machine learning. 
As a notable application, we present two streamlined protocols for the time-reversal simulation of unknown qubit unitary evolutions, reducing the ancilla qubit overhead from six to three compared to the previous best-known method. 
We also extend PQComb to solve the problems of qutrit unitary transformation and channel discrimination.
Furthermore, we demonstrate the hardware efficiency and robustness of our qubit unitary inversion protocol
under realistic noise simulations of IBM-Q superconducting quantum hardware, yielding a significant improvement in average similarity over the previous protocol under practical regimes. PQComb's versatility and potential for broader applications in quantum machine learning pave the way for more efficient and practical solutions to complex quantum tasks.
\end{abstract}

\date{\today}

\maketitle

\section{Introduction}
\begin{figure*}[t]
\centering
\includegraphics[width=0.8\textwidth]{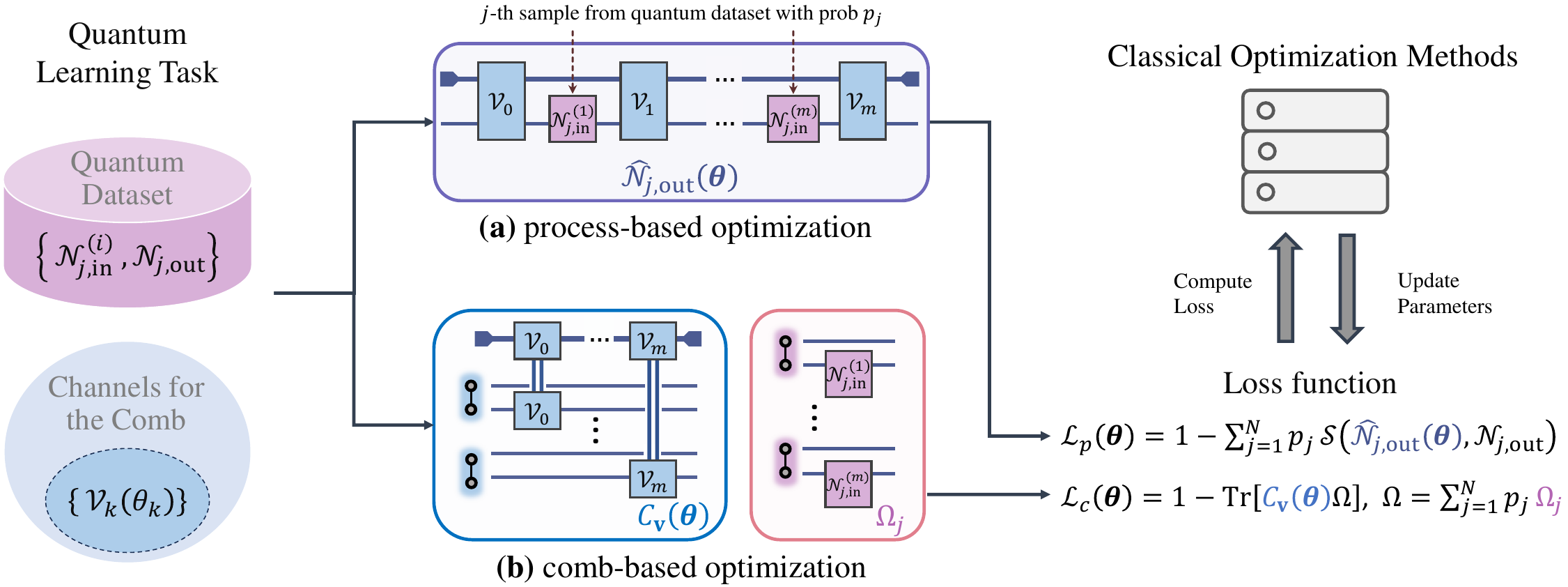}
\caption{ The overview of training formalism for the parameterized quantum comb framework. Within this scheme, the channel for the $k$-th tooth $\cV_k(\theta_k)$ is now parameterized by $\theta_k$ that remains tunable to adjustments during the optimization phase, and $\bm \theta = (\theta_0, \ldots, \theta_m)$ is denoted as the vector of all parameters in this PQComb. 
\textbf{(a)} describes how the PQComb trains the protocol using the process-based loss function $\cL_p$, which is computed by the average dissimilarity between the sampled output process $\widehat{\cN}_{j, \textrm{out}}(\bm \theta)$ and the expected process $\cN_{j, \textrm{out}}$. \textbf{(b)} describes how the PQComb trains the protocol using the comb-based loss function $\cL_c$, which optimizes the Choi operator of the circuit $C_\bmV(\bm \theta)$ using the performance operator $\O$. Here each pair of two dots connected by a line represents the unnormalized maximally entangled state. }
\label{fig:pqcomb}
\end{figure*}

In quantum computing, we are capable not only of transforming states but also of transforming processes. Designing quantum circuits to transform input operations has a wide range of applications in quantum computing, quantum information processing, and quantum machine learning. The networks that perform such transformations are known as super-channels, which take processes as inputs and output the corresponding transformed process. 

In general, such super-channel can be realized with a quantum circuit architecture~\cite{chiribella2008transforming, chiribella2008quantum}, namely a \emph{quantum comb}. One typical example is a quantum sequential comb as shown in Figure~\ref{fig:pqcomb} block (a), which takes quantum operations as sequential inputs and returns a new operation close to the target transformation. 
Quantum comb is widely applied in solving process transformation problems, including transformations of unitary operations such as inversion~\cite{quintino2019reversing, Yoshida2023}, complex conjugation~\cite{Miyazaki2019}, control-$U$ analysis~\cite{chiribella2016optimal}, as well as learning tasks~\cite{bisio2010optimal, sedlak2019optimal}. It can also be used for analyzing more general processes~\cite{zhu2024reversing} and has also inspired structures like the indefinite causal network~\cite{chiribella2013quantum, oreshkov2012quantum}. 

Previously, the approach to determine the quantum comb for target transformation is based on SDP using \Choi isomorphism, which takes the Choi operator of the quantum comb as the variable to obtain a feasible comb. Due to its guaranteed convergence, this method has been widely adopted and yields optimal quantum performance for specific tasks.
However, one major problem of the SDP method is that the dimension of the Choi operator grows exponentially with the number of comb slots, making it impossible to conduct numerical experiments for large-scale problems.
Additionally, as SDP ultimately returns the Choi operator of the quantum comb, deriving a physical implementation of the circuit, such as converting it into a standard circuit model, is far from straightforward.

Drawing from the transformative impact of deep learning in areas such as the game of Go~\cite{Silver2016} and protein folding prediction~\cite{Jumper2020}, 
we seek to leverage machine learning paradigms to enhance the
exploration of quantum information technologies. 
In particular, machine learning has been instrumental in refining quantum processor designs~\cite{Mavadia2017,Wan2017,Lu2017b,Niu2019} and manipulating quantum entanglement~\cite{Wallnofer2020,zhao2021practical}. Also, previous work~\cite{zhu2023quantum} examined the integration of quantum comb into a quantum auto-encoder within the context of classical cloud computing. 
In this work, we employ machine learning strategies to tackle the complexities associated with higher-order quantum information transformations. By utilizing Parameterized Quantum Circuits (PQCs), we aim to pioneer new frontiers in the field.

Parameterized quantum circuits, which form a building block of quantum machine learning models, offer a modular approach by decomposing a quantum circuit into quantum gates characterized by tunable parameters~\cite{Benedetti2019a}. This allows for an iterative optimization process, often employing gradient descent algorithms akin to those found in classical machine learning. This method is more suitable for near-term quantum devices and has been applied in various variational quantum algorithms and quantum machine learning~\cite{peruzzo2014variational,mcclean2016theory,cerezo2021variational}. Due to its structure and its optimization method being similar to classical neural networks, it is also referred to as Quantum Neural Networks. 

Based on this idea, we introduce a comprehensive framework named ``PQComb'', which utilizes PQCs to establish a general quantum comb structure. This framework is applied to the task of transforming quantum processes, where we model the transformation as a quantum comb and employ PQCs to represent the channels of each tooth within the comb. We approach the task as an optimization problem, leveraging classical optimization strategies to optimize the performance of the quantum circuit for the task specified. The optimization is done by adjusting the parameters within this network. Through this framework, we extend PQC into a broader and adaptive quantum neural network with memory to deal with higher-order transformation tasks and, in particular, to develop new protocols for unitary transformations.

To demonstrate the practical value of our methodology, we present experimental results on several problems. Notably, PQComb is applied to find an exact and deterministic protocol that performs arbitrary qubit-unitary inversion by querying four times the unitary itself, proven to be optimal in terms of gate usage~\cite{Yoshida2023}. It also advances the statement of art by reducing the number of auxiliary qubits required from six~\cite{Yoshida2023} to merely three. 
Furthermore, we conducted noise simulations that showcase the hardware efficiency of our circuit, which highlights the robustness of our approach in practical, noisy environments. 
It is worth noting that through the analysis of the circuit in detail, a protocol for achieving arbitrary dimension unitary inversion is developed~\cite{chen2024quantum}, which is the first deterministic and exact approach to reverse general unknown quantum time evolutions, resolving a long-standing fundamental problem.
For qutrit unitary transformation, we derive two near-exact and deterministic protocols that achieve qutrit-unitary inverse and transpose, by querying ten and seven times of the given gate respectively. It shows the ability of PQComb to address problems beyond the capabilities of the SDP method. 
Besides, we conducted experiments on channel discrimination, further showcasing the broad applicability of PQComb beyond process transformation.

%%%%%%%%%%%%%%%%%%%%%%%%%%%%%%%%%%%%%%%%%%%%%%%%%%%%%%%%%%%%%%%%%%%%%%%%%

\section{Results}

\subsection{The PQComb framework}~\label{sec:pqcomb framework}

Quantum comb can be classified into parallel and sequential types~\cite{chiribella2008quantum}, with the former structure being a special case of later one. As illustrated in Figure~\ref{fig:pqcomb} block (a), a sequential circuit involves a sequence of data processing operators $\cV_0, \ldots, \cV_m$ where each pair $\cV_j$ and $\cV_{j+1}$ shares a memory system. This arrangement adaptively transforms input processes $\Nin^{(1)}, \ldots, \Nin^{(m)}$ into an output process $\Nout$.
The quantum comb is noteworthy for its capacity to encapsulate the structure advanced by the quantum signal processing technique~\cite{low2016methodology, gilyen2019quantum}, an algorithmic framework that has been instrumental in unifying most well-known quantum algorithms~\cite{Martyn2021}. Furthermore, this architectural paradigm is also applicable to the data re-uploading model in quantum machine learning~\cite{perez-salinas2020data}, demonstrating that the Fourier features of a single-qubit quantum unitary can be learned by a data re-uploading QNN model~\cite{yu2022power}.

Studying quantum combs helps develop quantum protocols that simulate desired transformations. Mathematically, the goal of the process transformation is to design a quantum comb that outputs a target process with a sequence of input channels, which simulate the transformation
\begin{equation}
    f \left( \Nin^{(1)}, \ldots, \Nin^{(m)} \right) = \Nout
.\end{equation}
By taking the whole comb's Choi operator $C_\bmV$ as the variable, this problem is traditionally solved based on the SDP approach.  
The optimal comb is derived by maximizing the performance function $\trace{C_\bmV \O}$ under the comb's constraints, where $\O$ is the performance operator determined by the given input channels and the target output process~\cite{chiribella2016optimal}.
Although the SDP approach has a guaranteed convergence and allows for the determination of the Choi operator of a feasible quantum comb, the storage complexity of fully describing the Choi operator of a quantum comb with $m$ slots of dimension $d$ is at least $\cO(d^{4m})$.
This exponential growth makes the numerical processing of large-scale problems infeasible. 
Additionally, the practical compilation of such a Choi operator on actual quantum hardware is hindered by the prohibitive cost associated with non-restricted quantum resources — the infeasibility of constraining the ranks of channels within the convex optimization framework. 

By contrast, inspired by machine learning models, we introduce a PQC framework called the parameterized quantum comb (PQComb) to solve the process transformation problem, which addresses the two aforementioned challenges.
Specifically, we replace each data processing operator $\cV_k(\bm{\theta}_{k})$ by PQC, so that the whole comb is now characterized by all adjustable parameters, and the set of which is denoted as $\bm{\theta}$. Once a loss function is formulated, the parameter set is iteratively updated through classical optimization methods to obtain the circuit that yields the near-optimal protocol for the given task.

In general, we could choose the loss function to be the dissimilarity between the real output $\widehat{\cN}_{\textrm{out}}({\bm \theta})$ and the expected output $\Nout$, denoted as $1 - \cS(\widehat{\cN}_\textrm{out}({\bm \theta}), \Nout)$ for some computable similarity function $\cS$ between processes. 
Once the input channels are fixed, the output process can be obtained by matrix computation directly. 
We could note that, in contrast to the optimized values in the SDP approach which need to be linear functions, the PQC approach allows us to use nonlinear similarity functions. 
Optimizing this loss function will provide us with a practical solution to achieve the desired transformation. 
For the general scenario where the input processes are not fixed but sampled from operation sets, the loss function becomes the average of the dissimilarity, namely the \emph{process-based loss function}
\begin{equation}~\label{eqn:process loss}
    \cL_p({\bm \theta}) =  1 - \sum_{j = 1}^N p_j \cS\left(\widehat{\cN}_{j, \textrm{out}}({\bm \theta}), {\cN}_{j, \textrm{out}}\right)
,\end{equation}
where $\widehat{\cN}_{j, \textrm{out}}({\bm \theta})$ is the real output process for the $j$-th input combination with sample probability $p_j$, and ${\cN}_{j, \textrm{out}}$ is the expected output for this sample result. 
As an example, in unitary transformation tasks, the input channel in each slot is usually an unknown unitary gate selected randomly in Haar measure. 

As a parameter optimization method, we further propose two techniques that can accelerate training in specific scenarios by leveraging the unique properties of comb. The first technique pertains to the computation of the loss function. 
The computation of $\cL_p$ may need to perform sampling and matrix computation to cover all possible selected input channels in each iteration, which encounters diminished training efficacy when the set volume increases. In this case, if the similarity function can be expressed as a linear equation in terms of the PQComb's Choi operator $C_\bmV({\bm \theta})$ as 
\begin{equation}~\label{eqn:similarity}
    \cS\left(\widehat{\cN}_{j, \textrm{out}}({\bm \theta}), {\cN}_{j, \textrm{out}}\right) = \trace{C_\bmV({\bm \theta}) \O_j} \, ,
\end{equation}
where $\O_j$ is the performance operator determined by $\cN_{j, \textrm{in}}^{(1)}$, $\ldots$, $\cN_{j, \textrm{in}}^{(m)}$ and $\cN_{j, \textrm{out}}$~\cite{quintino2022deterministic}, we propose an alternative loss function to overcome the sampling problem, namely the \emph{comb-based loss function}
\begin{equation}~\label{eqn:comb loss}
    \cL_c({\bm \theta}) = 1 - \trace{C_\bmV({\bm \theta})\O}
,\end{equation}
where $\O = \sum_{j = 1}^N p_j \O_j$.
This loss function incorporates the features of both PQC and quantum comb.
The Choi operator $C_\bmV({\bm \theta})$ can be calculated by inserting unnormalized maximally entangled states to all input systems of the parameterized comb. Since the performance operator $\O$ is determined by the input channels and the expected output, it allows for pre-computation, thus avoiding the need for sampling at every iteration.

The second technique is about an initialization scheme for the parameters $\bm{\theta}$, called the \emph{SWAP-based optimization method}, which is particularly effective when dealing with large slot numbers. It is readily apparent that as the number of slots increases, initializing the parameters ${\bm \theta}_{\textrm{ini}}$ randomly can result in a poor initial value of the loss function. This not only prolongs the overall training process but also increases the likelihood of encountering local minima and other optimization issues. 
To address this problem, for a given slot, if we sequentially connect a $1$-slot circuit after it, where the operations on the two `teeth' correspond to SWAP operations between the ancilla and target systems, then regardless of the operation inserted into the last slot, the output process of the entire comb $\widehat{\cN}_{j, \textrm{out}}$ remains unchanged. Based on this observation, we can set the initial parameters of the $(m+1)$-slot comb by taking the trained parameters of the $m$-slot comb for the first $m+1$ teeth, and then, add two new teeth whose parameters are trained to be the SWAP gate. 
This will give a good initialization and significantly speed up the training process. 
Detailed optimization procedures are summarized in Supplementary Note 1~\footnote{See Supplementary Note 1 in the supplementary information.}.

In the next two subsections, we introduce several practical applications to showcase the value of the PQComb method. Most notably, it has enabled us to develop a protocol that perfectly implements the qubit unitary inversion, i.e., realizing $f(U^{(1)}, \ldots, U^{(m)}) = U^{-1}$. Compared to the previous protocol in~\cite{Yoshida2023}, the PQComb-derived protocol reduces the required number of ancilla qubits from six to three and simplifies the circuit implementation. Furthermore, this approach inspired the first algorithm capable of achieving unitary inversion in arbitrary dimensions deterministically and exactly~\cite{chen2024quantum}. We will first introduce the task of unitary inversion, followed by a detailed explanation of how the final protocol is obtained using PQComb. The performance of the proposed protocol is further highlighted under various noise models. Additionally, we also explore the applications in channel discrimination and qutrit unitary transformation, illustrating that PQComb has broad potential beyond process transformation and can handle problems involving larger slot numbers, which are numerically intractable using the SDP approach.

\subsection{Qubit-unitary inversion}~\label{sec:2-unitary inv}

The time evolution of a closed quantum system can be characterized by a unitary operator $U = e^{-iHt}$ with a Hamiltonian $H$ and time $t$. One can always reverse this transformation via the inverse operation $U^{-1} = e^{iHt}$.
The reversible nature of quantum unitary reveals a fundamental distinction between quantum computing and classical computing, which also mirrors the time-reversal symmetry of the underlying quantum mechanics.

The simulation of time-reversed quantum unitary evolution is not only a conceptual cornerstone in the realm of quantum information~\cite{Aharonov1990}, but it also serves as a key technology for the manipulation of quantum systems. This intricate process is pivotal for measuring out-of-time-order correlators~\cite{Maldacena2016,Li2017e,Garttner2017}, which serve as diagnostics for quantum chaos and entanglement dynamics. Moreover, the ability to reverse an unknown unitary evolution is an important building block for quantum algorithms (e.g., quantum-signal-processing-based algorithms~\cite{low2017optimal, gilyen2019quantum, wang2023quantum}), underscoring its significance in advancing quantum computational capabilities.

\begin{figure*}[t]
\centering
\[
\Qcircuit @C=1em @R=1em {
    & \lstick{\ket{0}} & \qw
& \qw & \multigate{2}{G}
& \gate{X}   & \multigate{2}{G^\dag} &\qw
& \qw & \multigate{2}{G}  
& \gate{X} & \multigate{2}{G^\dag}  
& \qw & \qw & \qw
& \rstick{\ket{0}} \qw\\
    & \lstick{\ket{0}} & \gate{H}
& \multigate{2}{\cQ_{\Uin}} & \ghost{G}
& \multigate{2}{\cQ_{\Uin}}  & \ghost{G^\dag}
& \multigate{2}{\cQ_{\Uin}} & \ctrlo{1} & \ghost{G}
& \multigate{2}{\cQ_{\Uin}} & \ghost{G^\dag}
& \gate{H} & \qw & \ctrl{2}
& \qw \backslash \\
    & \lstick{\ket{0}} & \gate{H} 
& \ghost{\cQ_{\Uin}} & \ghost{G}
& \ghost{\cQ_{\Uin}}  & \ghost{G^\dag} 
& \ghost{\cQ_{\Uin}} & \gate{-Z} & \ghost{G} 
& \ghost{\cQ_{\Uin}} & \ghost{G^\dag} 
& \gate{H} & \ctrl{1} & \qw
& \qw \backslash \\
    & \lstick{\ket{\varphi}} & \qw 
& \ghost{\cQ_{\Uin}} & \qw 
& \ghost{\cQ_{\Uin}} & \qw 
& \ghost{\cQ_{\Uin}} & \qw & \qw
& \ghost{\cQ_{\Uin}} & \qw 
& \qw & \gate{Y} & \gate{X}
& \rstick{\Uin^{-1} \ket{\varphi}} \qw}
\]
\caption{The proposed unitary inversion protocol for arbitrary single-qubit unitary $\Uin$. One can use 3 ancilla qubits and 4 queries of $\Uin$ to realize qubit-unitary inversion. Note that the output state of the first ancilla qubit will be a zero state without post-selection. The implementations of $\cQ_{\Uin}$ and $G$ are deferred to Supplementary Note 2~\cite{Note2}.}
\label{fig:4 3}
\end{figure*}
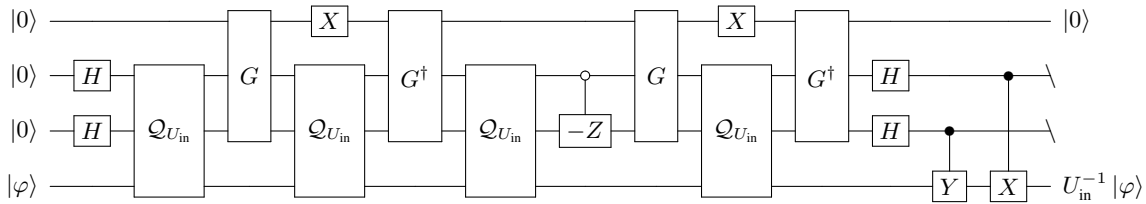

Reversing an unknown unitary evolution presents a notable challenge since it typically requires complete knowledge of the system but the information of a physics system in nature is often beyond our grasp. For the implementation of the inverse operation, one must have an exact characterization of the unitary transformation or the underlying Hamiltonian. However, quantum process tomography, the standard technique for such characterization of unitary operation, demands an impractically large number of measurements to fully describe a quantum process~\cite{Baldwin2014,Gutoski2014,Mohseni2008,Haah2023}. This requirement renders the exact reversal of a general unknown unitary operation impractical, as the conventional approach of learning and inverting is resource-prohibitive. 

While process tomography is challenging, simulating the unitary inverse $U^{-1}$ using the original unitary operation $U$ is still possible.
Higher-order transformations of quantum dynamics provide a potentially feasible approach for transforming an unknown unitary to its inverse. 
In particular, Ref.~\cite{Quintino2019, yang2021representation, trillo2023universal} introduced probabilistic universal quantum algorithms that execute the exact inversion of an unknown unitary operation. Ref.~\cite{Yoshida2023} further established the first deterministic and exact protocol for reversing any unknown qubit-unitary operations based on the SDP approach. 
To numerically handle the program with 4 slots, they imposed specific symmetry conditions on the Choi operator of the comb, resulting in a circuit with at least six ancilla qubits. However, whether these symmetry conditions are necessary, that is, whether the number of ancilla qubits can be further reduced, remains an open question.

\subsubsection{Deterministic and exact protocols by PQComb}

In this subsection, we address the problem by applying PQComb to the task and present a unitary inversion protocol that uses only three ancilla qubits. 
Here we denote two systems in this structure: the \emph{main system}, where the input unitary $\Uin$ operates, and the \emph{ancilla system}, for other qubits. The main system accepts an arbitrary state $\ket{\varphi}$ as input and is expected to output $\Uin^{-1} \ket{\varphi}$. The ancilla system, consisting of $n_a$ qubits, starts in the zero state and would be traced out at the end of the quantum comb.

For this task, we choose the comb-based loss function $\cL_c$ to train parameters of the circuit, where the performance operator $\O$ in Equation~\eqref{eqn:comb loss} is
\begin{align}
    \O
    \approx \frac{1}{N} \sum_{j = 1}^N \ketbravec{U_j^{-1}}{U_j^{-1}} \ox \ketbravec{\overline{U_j}}{\overline{U_j}}^{\ox m} \, .
\end{align}
Here $\kett{U} = \sum_k (U\otimes I) \ket{k} \ket{k}$ corresponds to the Choi operator of unitary gate $U$, and the set $\set{U_j}_{k=1}^N$ is randomly sampled from the special unitary group $\su(2)$ with size $N = 10^4$. 
The ansatz we used is shown in Supplementary Note 2~\footnote{See Supplementary Note 2 in the supplementary information.}.
One can then follow the optimization procedure in Figure~\ref{fig:pqcomb} to experimentally find the protocol with the optimal loss function for each setting $(m, n_a)$, as summarized in Table~\ref{tab:2-unitary inv}. 
Notably, the average similarity obtained by the PQComb matches the optimal value for $1 \leq m \leq 5$ within a tolerance of $1\cdot 10^{-3}$~\cite{Yoshida2023}.

\begin{table}[htbp]
\centering 
\caption{ The table summarizing the maximal average similarity obtained by the PQComb for qubit-unitary inversion task, under different pairs of $(m, n_a)$. Here $m$, $n_a$ are the number of slots and ancilla qubits, respectively; $\infty$ refers to the case when the number of ancilla is unlimited, where the optimal value is given by SDP~\cite{Yoshida2023}. }
\label{tab:2-unitary inv}
\begin{tabular}{c|cccccc}
\toprule
$m \backslash n_a$ & 0 & 1 & 2 & 3 & {$\cdots$} & $\infty$ \\
\midrule
\addlinespace
1 & {0.500} & \multicolumn{4}{r}{$\xrightarrow{\hspace*{8.5em}}$} & 0.500  \\
\addlinespace
2 & \lowlight{0.34} & \lowlight{0.63} & \lowlight{0.73} & {0.750} & {$\xrightarrow{\hspace*{1em}}$} & 0.750 \\
\addlinespace
3 & \lowlight{0.37} & \lowlight{0.69} & {0.933} & \multicolumn{2}{r}{$\xrightarrow{\hspace*{3.5em}}$} & 0.933 \\
\midrule
\addlinespace
4 & \lowlight{0.48} & \lowlight{0.78} & \lowlight{0.96} & \highlight{0.999} & {$\xrightarrow{\hspace*{1em}}$} & {1} \\
\addlinespace
5 & \lowlight{0.72} & \lowlight{0.82} & \highlight{0.999} & \multicolumn{2}{r}{$\xrightarrow{\hspace*{3.5em}}$} & {1}  \\
\addlinespace
\bottomrule
\end{tabular}
\end{table}

Table~\ref{tab:2-unitary inv} indicates that by utilizing three ancilla qubits and four queries of the unitary operator $\Uin$, PQComb is capable of providing a near-exact and deterministic protocol to approximate $\Uin^{-1}$.
After achieving this protocol, we further refine our optimization based on the current structure, which leads to a more streamlined training ansatz that suppresses the average dissimilarity to $10^{-6}$, and eventually resulting an exact and deterministic protocol illustrated in Theorem~\ref{thm:4 3}. More details for ansatz selection and refinement are deferred to Supplementary Note 2~\cite{Note2}.

\begin{theorem}[3-ancilla 4-call Protocol]~\label{thm:4 3}
    There exists a quantum circuit implementing $\Uin^{-1}$ by 3 ancilla qubits and 4 calls of a single-qubit unitary $\Uin$, such that
\begin{equation}
\begin{aligned}
    & \quad \trace[23]{\fourcir(\Uin) \left(\ketbra{000}{000} \ox \rho \right) \fourcir(\Uin)^\dag} \\
    &= \ketbra{0}{0}_1 \ox \Uin^{-1} \rho \Uin
,\end{aligned}
\end{equation}
where $\fourcir(\Uin)$ gives the unitary matrix of the output process.
\end{theorem}

\noindent\textit{Sketch of Proof.} For $\Uin \in \su(2)$, a decomposition on Pauli basis is $\Uin = \cos(\theta / 2) I - i\sin(\theta / 2) \Vec{n} \cdot \Vec{\sigma} $, with $\Vec{n} = (n_x, n_y, n_z)$ respective to the coefficients of Pauli operators. Then the output state of the circuit in Figure~\ref{fig:4 3} is
\begin{equation}
\begin{aligned}
\dfrac{1}{2} \ket{0} \ox ( 
&\cos{\dfrac{\theta}{2}} \ket{00}  - i\sin{\dfrac{\theta}{2}} n_{y} \ket{01} - \\
&i\sin{\dfrac{\theta}{2}} n_{x} \ket{10}- i\sin{\dfrac{\theta}{2}} n_{z} \ket{11} ) \ox \Uin^{-1} \ket{\varphi}
\end{aligned}
\end{equation}
and hence, the statement follows. More details are deferred to Supplementary Note 2~\cite{Note2}.

Inspired by the ansatz we obtained, we did further numerical experiments and discovered a circuit that deterministically and exactly implements the qubit unitary inversion querying $\Uin$ five times. In this protocol, all ancilla qubits are reset to $\ket{0}$ after the circuit execution. The significance of this finding is that it directly inspired an algorithm for achieving unitary inversion in arbitrary dimensions, thereby addressing a long-standing open problem~\cite{chen2024quantum}. The detailed circuit implementing this approach is presented in Supplementary Note 2~\cite{Note2} and is summarized in the following corollary.

\begin{corollary}[3-ancilla 5-call Protocol]~\label{coro:5 3}
    There exists a quantum circuit implementing $\Uin^{-1}$ by 3 ancilla qubits and 5 calls of a single-qubit unitary $\Uin$, such that
\begin{equation}
    \fivecir(\Uin) \ket{000, \psi} = \ket{000} \ox \Uin^{-1} \ket{\psi}
,\end{equation}
    where $\fivecir(\Uin)$ gives the unitary matrix of the output process.
\end{corollary}
For the sake of clarity and differentiation, we refer to the protocol in Theorem~\ref{thm:4 3} as the ``4-call protocol'' and to that in Corollary~\ref{coro:5 3} as the ``5-call protocol''.

\subsubsection{Noise simulation of qubit-unitary inversion protocols}~\label{sec:noise}

In the noisy intermediate-scale quantum (NISQ) era, devices are inevitably affected by noise, underscoring the necessity of evaluating the performance of quantum algorithms under realistic noise conditions. Given this context, it is crucial to evaluate the robustness of our proposed unitary inversion protocols under practical devices.
We simulated the performance of our entire circuit under realistic noise conditions by utilizing the IBM-Q cloud service. 
Our results showcase the performance of our protocols compared to the previous approach~\cite{Yoshida2023}, attributable to the reduced circuit width and depth facilitated by our more compact constructions. 
 
We consider the scenario where our entire circuit is affected by real-device noise. This simulation is based on the IBM-Q cloud service, with noise settings from five different IBM quantum devices. Under the same noisy model, both protocols demonstrate superior performance compared to the protocol introduced in Ref.~\cite{Yoshida2023}. This improvement can be attributed to the fact that our protocols have halved the number of ancilla qubits and reduced the compiled depth by a factor of five. These optimizations underscore the efficiency of these two protocols, showcasing PQComb as a hardware-efficient algorithmic designer for practical quantum devices. 

\begin{figure}[H]
\centering
\hspace{0.2em}
\subfloat{
\includegraphics[width=\linewidth]{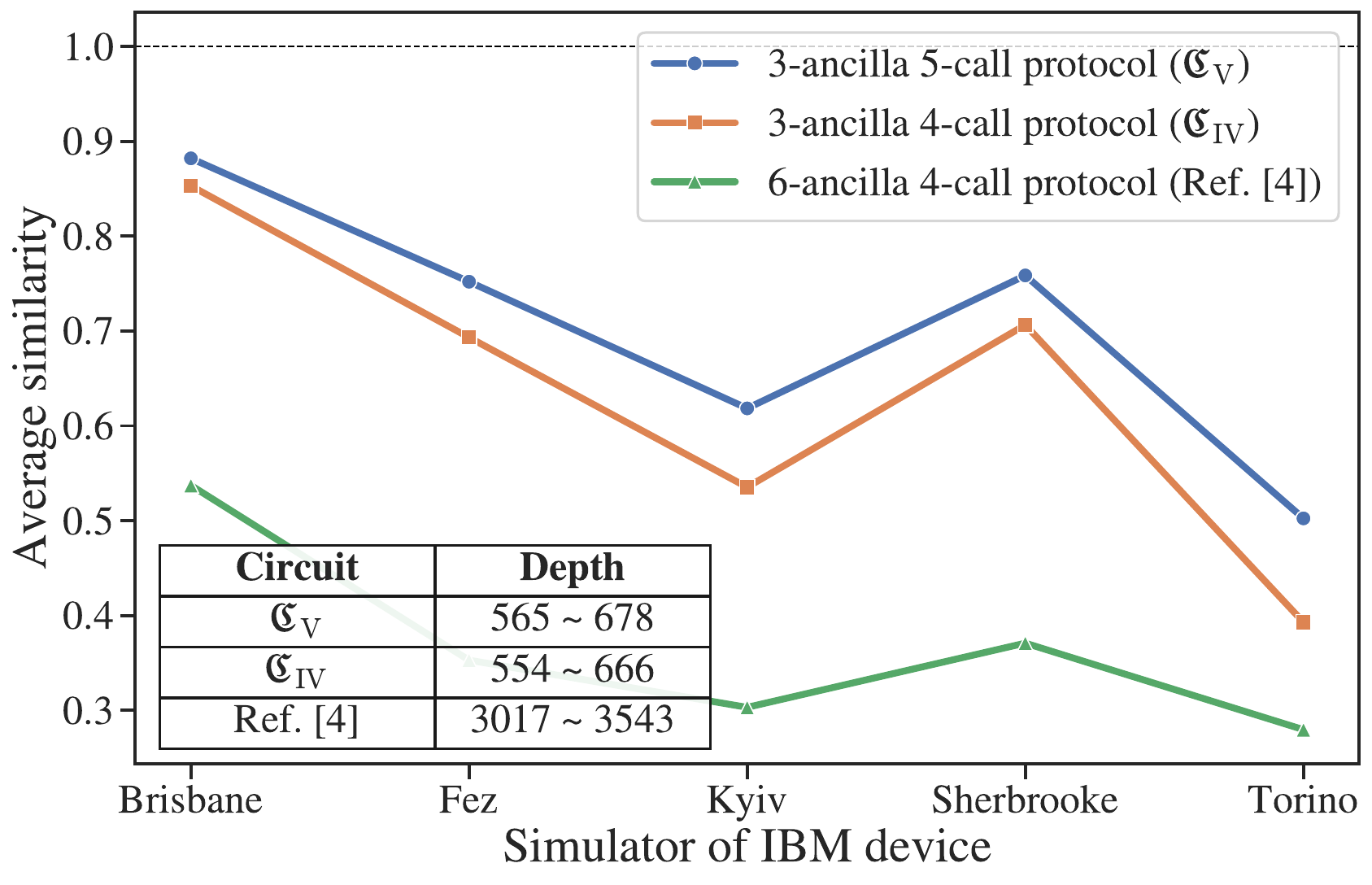}
}
\caption{Simulation of our two protocols and the previous protocol in Ref.~\cite{Yoshida2023} under the noise settings of five real quantum devices. Here we refer to the protocols in Theorem~\ref{thm:4 3} and Corollary~\ref{coro:5 3} as the ``4-call protocol'' and ``5-call protocol'', respectively.}
\label{fig:ibm noise}
\end{figure}

It is interesting to note that, the 5-call protocol demonstrates higher average similarities than the 4-call protocol across all these practical settings derived from real quantum devices. We guess this is because the 5-call protocol reset all three ancilla qubits to zero states, making it a clean protocol that all four qubit systems are decoherent from one another and hence be more robustness to experimental noise.
It is also worth noting that in this simulation our circuit has not yet been optimized for architecture. Optimizing at the circuit level may further enhance the performance of our protocol under noise conditions. The details regarding this simulation experiment can be found in Supplementary Note 2~\cite{Note2}.

\subsection{Other applications}~\label{sec:other app}

In addition to the task of qubit unitary inversion, we also applied PQComb to the tasks of qutrit unitary transformation and channel discrimination to demonstrate its broader applicability. The experimental results are available to our GitHub repository~\cite{coderepo}.

\subsubsection{Qutrit-unitary transformations}

For qutrit unitary transformations, we specifically focused on the problems of simulating the transpose and inverse of the input unitary. Notably, due to the larger number of queries required, the SDP approach faces the memory issue and cannot solve these problems. 

To address this, we employed the \emph{SWAP-based optimization method} to initialize the parameters and trained the PQComb using the \emph{process-based loss function}
\begin{equation}
    \cL_p(\bm \theta) = 1 - \frac{1}{3N} \sum_{j=1}^N   \langle\!\langle f(U_j) \lvert \, \cJ_{j, \textrm{out}}(\bm \theta) \, \lvert f(U_j) \rangle\!\rangle
,\end{equation}
where $\cJ_{j, \textrm{out}}$ is the Choi operator of the $j$-th output process and $f(U) = U^T \text{ or } U^{-1}$ is the target transformation. As system dimensions grow, we use $N = 10^4$ samples to train the circuit and test it with additional $10^5$ samples to evaluate protocol performance.
We derived near-perfect circuits for both transpose and inverse by using the input unitary seven and ten times respectively, each achieving a test average fidelity above $0.99$. Some of these numerical results are summarized in Table~\ref{tab:qutrit}.

\begin{table}[H]
\centering
\caption{Two tables summarizing the maximal average similarity obtained by the PQComb for qutrit-unitary inversion and transpose, respectively, under different pairs of $(m, n_a)$. Here $m$, $n_a$ are the number of slots and ancilla qutrits, respectively; $\infty$ refers to the case when the number of ancilla is unlimited, where the optimal value is given by SDP~\cite{Yoshida2023}.}\label{tab:qutrit}
\subfloat[$f(U) = U^{-1}$]{\resizebox{!}{0.065\textheight}{
\begin{tabular}{c|ccccc}
    \toprule
    $m \backslash n_a$& 1 & 2 & 3 & {$\cdots$} & $\infty$ \\
    \midrule
    1 & 0.222 & \multicolumn{3}{r}{$\xrightarrow{\hspace*{5.5em}}$} & 0.222 \\
    2 & 0.333 & \multicolumn{3}{r}{$\xrightarrow{\hspace*{5.5em}}$} & 0.333 \\
    3 & \lowlight{0.35} & \lowlight{0.39} & 0.429 & {$\mathrel{\leadsto}$} & 0.444 \\
    4 & \lowlight{0.35} & \lowlight{0.48} & 0.541 & {$\mathrel{\leadsto}$} & 0.556 \\
    5 & \lowlight{0.41} & \lowlight{0.53} & 0.664 & {$\mathrel{\leadsto}$} & 0.667\\
    10 & \lowlight{0.56} & \lowlight{0.73} & \highlight{0.995} & {$\mathrel{\leadsto}$} & $?$ \\
    \bottomrule
\end{tabular}
}}
\hspace{1em}
\subfloat[$f(U) = U^T$]{\resizebox{!}{0.065\textheight}{
\begin{tabular}{c|ccc}
    \toprule
    $m \backslash n_a$& 1 & 2 & 3 \\
    \midrule
    1 & \lowlight{0.22} & \lowlight{0.22} & \lowlight{0.22} \\
    2 & \lowlight{0.30} & \lowlight{0.37} & \lowlight{0.41} \\
    3 & \lowlight{0.31} & \lowlight{0.50} & \lowlight{0.60} \\
    4 & \lowlight{0.32} & \lowlight{0.63} & \lowlight{0.79} \\
    5 & \lowlight{0.41} & \lowlight{0.69} & \lowlight{0.91} \\
    7 & \lowlight{0.47} & \lowlight{0.85} & \highlight{0.994} \\
    \bottomrule
\end{tabular}
}}
\end{table}

For the unitary inversion task, when $m \leq 5$, our training results closely align with optimal fidelity obtained from the SDP method that utilizes symmetry conditions in the special unitary group $\su(3)$~\cite{Yoshida2023}. When further increasing the slot number, SDP methods become powerless, while our results show that qutrit-unitary inversion is nearly feasible by querying the input unitary ten times.
Additionally, these results are based on preliminary experiments with a universal ansatz, further refinement of the ansatz may lead to improved performance or fewer query numbers.

For unitary transpose, it is worth noting that our numerical results may provide insight into the problem discussed in~\cite{Odake2024}, where the authors derived lower bounds for simulating unitary inverse and transpose in arbitrary dimensions. Specifically, they showed that to realize the inverse requires at least $d^2$ queries, while the transpose may require only $\cO(d)$ queries. 
The only existing deterministic and exact high-dimensional protocol for unitary transpose was based on a variant of the unitary inverse protocol from~\cite{chen2024quantum}, which still requires $\cO(d^2)$ queries. Therefore, the exact query complexity for unitary transpose remained an open question. 
Our results lead to the conjecture that, simulating the transpose may be different from the inverse and could potentially be achieved in $\cO(d)$ queries. This insight, based on PQComb's numerical experiments, may inspire future research in this direction.

\subsubsection{Channel discrimination}
Additionally, we analyze the channel discrimination problem using PQComb. 
We note that channel discrimination is a quantum information task that distinguishes between two noise channels. Given finite copies of an unknown input channel selected from these two, the goal is to determine which channel it is.
For this task, we discriminate between two qubit channels: an amplitude damping channel $\cA$ and a bit flip channel $\cE$ with noise parameters 0.67 and 0.13, respectively. The task requires designing a quantum comb that produces binary output: 0 for channel $\cA$ and 1 for channel $\cE$. The discrimination performance is evaluated using a modified comb-based loss function:
\begin{align}
\cL_c(\bm \theta) = 1 &- \frac{1}{2} \trace{\bra{0}_\bmF C_{\bmV}(\bm \theta) \ket{0}_\bmF \cJ_\cA^{\ox m} } \\
&- \frac{1}{2}\trace{\bra{1}_\bmF C_{\bmV}(\bm \theta) \ket{1}_\bmF \cJ_\cE^{\ox m}}
,\end{align}
where $\ket{0}_\bmF, \ket{1}_\bmF$ represent the zero and one states in the final system, with identity operators omitted in other subsystems.

In~\cite{bavaresco2021strict}, the authors used the SDP approach to investigate this problem with $m=2$ and establish a strict hierarchy that the sequential protocol can strictly outperform any parallel protocol which queries the two channels simultaneously. 
Using our parameterized approach, we trained a 2-slot circuit with five ancilla qubits that achieved an average success probability of 0.8444 aligns with the previous result. 
As shown in~\cite{bavaresco2021strict}, parallel combs are limited to success probabilities below 0.844, the circuit we find exceeds this threshold. 
These findings illustrate the applicability of PQComb beyond unitary transformations and its ability to achieve results comparable to those obtained through SDP methods.

%%%%%%%%%%%%%%%%%%%%%%%%%%%%%%%%%%%%%%%%%%%%%%%%%%%%%%%%%%%%%%%%%%%%%%%%%%%%%%%%%%%%%%%%%%%%%
\section{Discussion}

In this work, we developed PQComb for exploring the capabilities of quantum combs in transforming quantum processes via the idea of supervised learning. Compared to the standard SDP method, this approach has the advantages in providing more flexible loss functions tailored to different tasks and resources, designing practical circuits for actual implementation, and exploring protocols beyond the computational limit of SDP problems.

One major contribution from our work is the creation of a more straightforward approach to reverse unknown qubit-unitary operations, derived from the PQComb's optimization for this specific task. The proposed new protocols simplify the circuit complexity and enhances the efficiency of qubit-unitary inversion by halving the circuit width. Such a reduction not only highlights the practical value of quantum comb structures but also illustrates PQComb's capacity for generating cutting-edge quantum protocols and algorithms. 
The hardware efficiency of our protocols is further demonstrated by noise simulations across various noise models. 
Here, we note that the detailed analysis of reversing qubit-unitary operations presented in Supplementary Note 2~\cite{Note2} enhances the understanding of the general unitary inversion task, which is subsequently extended to arbitrary dimensions in~\cite{chen2024quantum}. Together with the applications of qutrit-unitary inverse and transpose transformations and channel discrimination, PQComb is shown to be an effective, hardware-efficient, and versatile framework for designing practical quantum protocols.

As a PQC-based framework, the optimization methods discussed here can be integrated with the NISQ devices. By leveraging the power of quantum computing, PQComb may analyze problems that are hard to analyze classically. Future research directions include exploring ansatzes for quantum combs with different structures to solve various types of problems and developing novel quantum algorithms with the aid of PQComb.  
As PQComb is a highly adaptable framework, it can be applied to a wide range of quantum computing tasks by modifying the structure of the training dataset or the trainable ansatz. For example, by sampling the dataset from Clifford gates and restricting the ansatzes to be stabilizer circuits~\cite{gottesman1997stabilizer}, PQComb can be harnessed to investigate problems in fault-tolerant quantum computing. Alternatively, by tailoring the target transformation and the loss function, PQComb can be leveraged to train channel inversion~\cite{zhu2024reversing}, obtain transformations of Hamiltonian dynamics~\cite{quintino2022deterministic, odake2024higherorder, odake2023universal}, or estimate unknown parameters of quantum systems~\cite{liufullyoptimized}.
We believe the results in this paper could pave the way for the application of the parametrized quantum combs across quantum computing and machine learning domains, opening up new possibilities for future research and development in these fields.

%%%%%%%%%%%%%%%%%%%%%%%%%%%%%%%%%%%%%%%%%%%%%%%%%%%%%%%%%%%%%%%%
\section*{Acknowledgement}
Y. M., L. Z., and Y.-A. Chen contributed equally to this work. We would like to thank Chengkai Zhu, Xuanqiang Zhao and Yu Gan for their helpful comments. Code used in the numerical experiments is available on \url{https://github.com/QuAIR/PQComb-codes}. Other code used in this study is available from the corresponding authors upon reasonable request. This work was partially supported by the National Key R\&D Program of China (Grant No.~2024YFE0102500), the Guangdong Provincial Quantum Science Strategic Initiative (Grant No.~GDZX2403008, GDZX2403001), the Guangdong Provincial Key Lab of Integrated Communication, Sensing and Computation for Ubiquitous Internet of Things (Grant No.~2023B1212010007), the Quantum Science Center of Guangdong-Hong Kong-Macao Greater Bay Area, and the Education Bureau of Guangzhou Municipality.

%%%%%%%%%%%%%%%%%%%%%%%%%%%%%%%%%%%%%%%%%%%%%%%%%%%%%%%%%%%%%%%%%%%%%%%%%
% Bibliography
%\bibliographystyle{alpha}
% \bibliographystyle{unsrt}

% \bibliography{references}

%%%%%%%%% SUPPLEMENTAL MATERIAL %%%%%%%%%

\appendix

\vspace{3cm}
\onecolumngrid
\vspace{2cm}

\clearpage
\begin{center}
\Large{\textbf{Supplementary Information} \\ \textbf{
}}
\end{center}

\renewcommand{\theequation}{S\arabic{equation}}
% \numberwithin{equation}{section}
\renewcommand{\thesubsection}{\normalsize{Supplementary Note \arabic{subsection}}}
\renewcommand{\theproposition}{S\arabic{proposition}}
\renewcommand{\thedefinition}{S\arabic{definition}}
\renewcommand{\thefigure}{S\arabic{figure}}
\renewcommand{\theHfigure}{S\arabic{figure}}
\setcounter{equation}{0}
\setcounter{table}{0}
\setcounter{section}{0}
\setcounter{subsection}{0}
\setcounter{proposition}{0}
\setcounter{definition}{0}
\setcounter{figure}{0}

\section{More details on PQComb}~\label{appendix:pqcomb}

\paragraph{\textbf{Settings of PQComb}.---} The parameterized quantum comb follows the structure of a sequential quantum comb presented in Figure~\ref{fig:sequential comb}. The notation adopted in this work is presented as follows: As a whole, the process transformation that this comb does is denoted by the symbol $\mathfrak{C}$, characterized by its Choi operator $C_\bmV$, where $\bmV$ denotes the operator set of all tooth. Looking inside the comb, $\bmP$ and $\bmF$ denote the input and output systems of the main register. $I_k$ and $O_k$ correspond to the input and output systems for the operation in the k-th slot. Additionally, $\Nin^{(k)}$ and $\Nout$ denote the $k$-th input process and the final output process highlighted by red boxes, respectively. Each tooth within the comb, corresponding to the $k$-th position, is characterized by a quantum operator $\cV_k$. Each two neighbor teeth $\cV_{k - 1}$ and $\cV_{k}$ share the aniclla system $\aux_k$. At the end of the circuit, the projector $\Pi$ determines the post-selection of the output auxiliary system. For the applications discussed in this paper, $\Pi$ is selected as the identity matrix, i.e., no post-selection happens and the output auxiliary system is just traced out. 

\begin{figure}[H]
    \centering
    \includegraphics[width=0.7\textwidth]{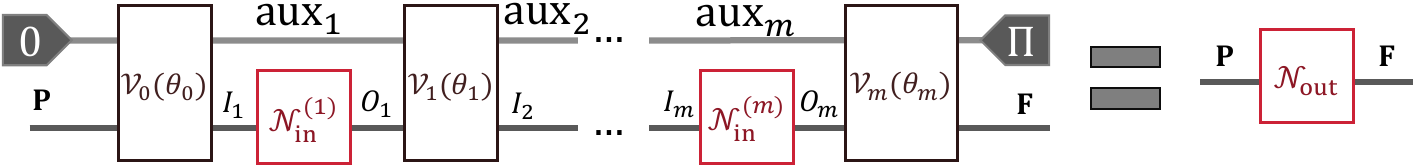}
    \caption{The schematic representation of a sequential quantum comb $\mathfrak{C}$.}
    \label{fig:sequential comb}
\end{figure}

Here we summarized the general optimization procedure for PQComb. See Figure 1 for more intuitive explanations.

\begin{algorithm}[H]
\KwIn{a quantum dataset $\cJ = \set{ \cN_{j, \textrm{in}}^{(k)} }_{j, k}$, a target transformation $f$ that outputs the expected process $\cN_{j, \textrm{out}} = f\left( \cN_{j, \textrm{in}}^{(1)}, \ldots, \cN_{j, \textrm{in}}^{(m)} \right)$, a similarity function $\cS$ that measures the similarity of two transformations.}
\KwOut{ a quantum sequential comb $\mathfrak{C}$ such that $\mathfrak{C}\left( \cN_{j, \textrm{in}}^{(1)}, \ldots, \cN_{j, \textrm{in}}^{(m)} \right) \approx \cN_{j, \textrm{out}}$ for all $\cN_{j, \textrm{in}}^{(k)} \in \cJ$.}

Depending on the size of dataset and characteristic of $f$, determine the $m$-slot sequential comb parameterized by trainable parameters $\bm \theta$, and the size $N$ of training set $\set{ \cN_{j, \textrm{in}}^{(k)} }_{j, k = 1}^{N, m}$. Denote the Choi of this parameterized comb as $C_{\bmV}(\bm \theta)$ \;

If the similarity function $\cS$ is a linear function by priori knowledge and $m$ is not large (depends on the classical computer), choose the comb-based optimization, i.e., the loss function is $\cL(\bm \theta) = 1 - \trace{C_\bmV({\bm \theta})\O}$\;

Otherwise, choose the process-based optimization, i.e., the loss function is $\cL(\bm \theta) =  1 - \sum_{j = 1}^N \cS(\widehat{\cN}_{j, \textrm{out}}({\bm \theta}), {\cN}_{j, \textrm{out}}) / N$ \;
Use classical optimization methods such as gradient descent method to obtain the optimal parameters ${\bm\theta}^* = \argmin_{\bm\theta} \cL(\bm \theta)$.  Then the quantum comb $C = C_{\bmV}({\bm\theta}^*)$ can obtain the minimum average dissimilarity with respect to the training set. Return $C$.
\caption{General optimization procedure in PQComb}
\label{alg:general}
\end{algorithm}
\vspace{2em}

\update{

\noindent Note that when one chooses the comb-based optimization, the performance operator can be constructed before optimization, as
\begin{equation}~\label{eqn:omega general}
    \O \approx \frac{1}{N} \sum_{j = 1}^{N} \left(\cJ_{j, \textrm{out}}\right)_{\bmP, \bmF} \ox \left(\cJ_{j, \textrm{in}}^{(1)}\right)_{I_1, O_1} \ox \ldots \ox \left(\cJ_{j, \textrm{in}}^{(m)}\right)_{I_m, O_m}
,\end{equation}
where $\cJ_{j, \textrm{out}}$ and $\cJ_{j, \textrm{in}}^{(k)}$ are the Choi representations of channels $\cN_{j, \textrm{out}}$ and $\cN_{j, \textrm{in}}^{(k)}$, respectively. 

\vspace{2em}
\paragraph{\textbf{SWAP-based optimization}.---} 
When the number of slots increases, initializing the parameters ${\bm \theta}_{\textrm{ini}}$ of the whole comb randomly may result in a poor initial value of the loss function. This not only prolongs the overall training process but also increases the likelihood of encountering local minima and other optimization issues. 
To better train PQComb with large slot numbers, we employ the SWAP-based optimization method. 
Before introducing how to initialize the parameters of an $(m+1)$-slot comb based on the trained parameters of an $m$-slot comb, we clarify the setting: 

\begin{enumerate}
    \item We consider a quantum sequential comb denoted by $\mathfrak{C}_m$, which consists of $m$ slots with dimension $d$. The ancilla system associated with this comb has dimension $d_a$ such that $d_a / d \in \NN$. 
    \item For $\mathfrak{C}_m$, we denote its parameters as ${\bm \theta}_{m}$, with ${\bm \theta}_{m}^\textrm{ini}$ be its initial value and ${\bm \theta}_{m}^\textrm{train}$ be the value after training. 
    \item Getting the $(m+1)$-slot comb from the $m$-slot comb is by sequentially connecting a parameterized 1-slot comb with the same ancilla system. The parameters of this 1-slot comb is denoted as ${\bm \theta}_{\textrm{con}}$
\end{enumerate}

Initially we will train this 1-slot comb $\mathfrak{C}_\textrm{SWAP}$ to get its initial parameters ${\bm \theta}_{\textrm{con}}^{\textrm{ini}}$ such that it performs like swaps one qudit in the ancilla system with the qudit in the main system, as shown below.

\begin{figure}[H]
    \centering
    \[
    \Qcircuit @C=1.5em @R=1.5em {
        & \lstick{\textrm{ancilla system of dimension } d_a / d} & \qw & \qw & \qw & \qw  \\
        & \lstick{\textrm{ancilla system of dimension } d} & \multigate{1}{\textrm{SWAP}_d} & \qw & \multigate{1}{\textrm{SWAP}_d} & \qw \\
        & \lstick{\textrm{Input state}} & \ghost{\textrm{SWAP}_d} & \gate{\textrm{Input slot}} & \ghost{\textrm{SWAP}_d} & \rstick{\textrm{Output state}} \qw
    }
    \]
    \caption{Circuit implementation of $\mathfrak{C}_\textrm{SWAP}$, where $\textrm{SWAP}_d$ is the generalized $d^2$-dimensional SWAP gate}
    \label{fig:CSWAP}
\end{figure}
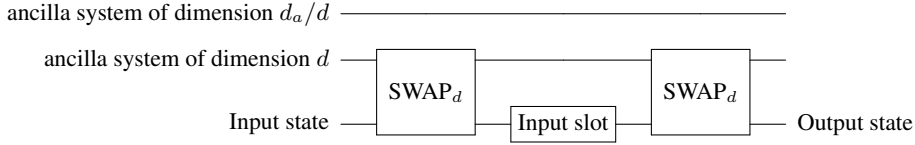

Based on this, we construct the $(m+1)$-slot comb by connecting the $m$-slot comb with this 1-slot comb as shown in Figure~\ref{fig:algorithm2}. 
Thus ${\bm \theta}_{m+1}$ consists of ${\bm \theta}_{m}$ and ${\bm \theta}_{\textrm{con}}$, which is initialized to be ${\bm \theta}_{m}^{\textrm{train}}$ and ${\bm \theta}_{\textrm{con}}^{\textrm{ini}}$. 
As the last 1-slot comb now performs as the identity channel on the main system, the loss function will start with the value obtained from training the $m$-slot comb. 

\begin{figure}[H]
    \centering
    \includegraphics[width=\linewidth]{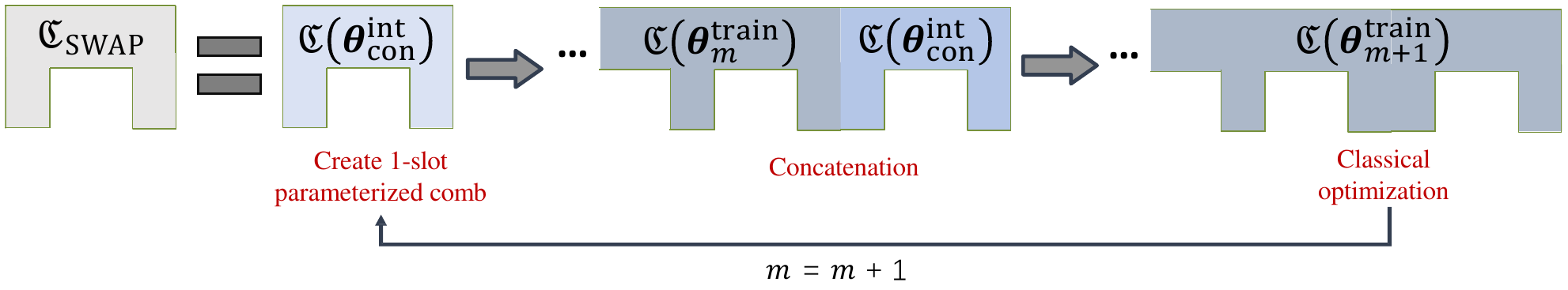}
    \caption{Graphical illustration of STEP 6 - 8 in Algorithm~\ref{alg:swap}.}
    \label{fig:algorithm2}
\end{figure}

This method can be used to train an $M$-slot protocol for some given tasks, starting from training a 1-slot comb. The detailed training process is shown in Algorithm~\ref{alg:swap}.

\begin{algorithm}[H]
\KwIn{A given transformation task as defined in Algorithm~\ref{alg:general}, and the maximal slot number $M$.}
\KwOut{A quantum sequential comb $\mathfrak{C}_M^*(\theta_M)$.}

Construct the 1-slot comb $\mathfrak{C}_{\textrm{con}}$ with parameter ${\bm \theta}_{\textrm{con}}$ for connection as shown in Figure~\ref{fig:algorithm2}\;
Train ${\bm \theta}_{\textrm{con}}$ to make it perform close to $\mathfrak{C}_\textrm{SWAP}$ as shown in Figure~\ref{fig:CSWAP}\;
Set $m=1$ and construct a 1-slot comb $\mathfrak{C}_1(\theta_1)$\;
Train $\theta_1$ based on Algorithm~\ref{alg:general} for the given transformation task to get the final parameters $\theta_1^\textrm{train}$\;
\While{$m < M$}
{
Construct the $(m+1)$-slot comb $\mathfrak{C}_{m+1}(\theta_{m+1})$ by connecting the $m$-slot comb $\mathfrak{C}_{m}(\theta_{m})$ with $\mathfrak{C}_\textrm{SWAP}$\;
Initialize $\theta_{m+1}^{\textrm{ini}}$ to be ${\bm \theta}_{m}^{\textrm{train}}$ and ${\bm \theta}_{\textrm{con}}^{\textrm{ini}}$\;
Train $\theta_{m+1}$ based on Algorithm~\ref{alg:general} for the given transformation task to get the final parameters $\theta_{m+1}^\textrm{train}$
}
Output $\mathfrak{C}^* = \mathfrak{C}_m$;
\caption{SWAP-based optimization procedure in PQComb}
\label{alg:swap}
\end{algorithm}
}

\clearpage
\section{Qubit-unitary inversion task}~\label{appendix:2-unitary inv}

\paragraph{\textbf{Detailed proof for Theorem 1}.---} In this subsection, we give the detailed proof for Theorem 1, and the circuit for the ``4-call protocol'' and the ``5-call protocol'' are shown below.

\begin{figure}[htbp]
\centering
\subfloat[]{\label{fig:4 3 with label}
\resizebox{0.85\textwidth}{!}{
\Qcircuit @C=1em @R=1em {
    & \lstick{\ket{0}_1} & \qw
& \qw & \multigate{2}{G}\barrier[0em]{3} 
& \qw \raisebox{3.5em}{$\ket{\Psi_{\text{I}}}$}
&  \gate{X}   & \multigate{2}{G^\dag}\barrier[0em]{3} 
& \qw \raisebox{3.5em}{$\ket{\Psi_{\text{II}}}$}  &\qw
& \qw & \multigate{2}{G} \barrier[0em]{3} 
& \qw \raisebox{3.5em}{$\ket{\Psi_{\text{III}}}$}  
& \gate{X} & \multigate{2}{G^\dag} \barrier[0em]{3} 
& \qw \raisebox{3.5em}{$\ket{\Psi_{\text{IV}}}$}  
& \qw & \qw & \qw
& \rstick{\ket{0}_1} \qw\\
    & \lstick{\ket{0}_2} & \gate{H}
& \multigate{2}{\cQ_{\Uin}} & \ghost{G} & \qw 
& \multigate{2}{\cQ_{\Uin}}  & \ghost{G^\dag} & \qw
& \multigate{2}{\cQ_{\Uin}} & \ctrlo{1} & \ghost{G} & \qw
& \multigate{2}{\cQ_{\Uin}} & \ghost{G^\dag} & \qw
& \gate{H} & \qw & \ctrl{2}
& \qw \backslash \\
    & \lstick{\ket{0}_3} & \gate{H} 
& \ghost{\cQ_{\Uin}} & \ghost{G} & \qw 
& \ghost{\cQ_{\Uin}}  & \ghost{G^\dag} & \qw
& \ghost{\cQ_{\Uin}} & \gate{-Z} & \ghost{G} & \qw
& \ghost{\cQ_{\Uin}} & \ghost{G^\dag} & \qw
& \gate{H} & \ctrl{1} & \qw
& \qw \backslash \\
    & \lstick{\ket{\varphi}} & \qw 
& \ghost{\cQ_{\Uin}} & \qw & \qw
& \ghost{\cQ_{\Uin}} & \qw & \qw
& \ghost{\cQ_{\Uin}} & \qw & \qw & \qw
& \ghost{\cQ_{\Uin}} & \qw  & \qw
& \qw & \gate{Y} & \gate{X}
& \rstick{\Uin^{-1} \ket{\varphi}} \qw}
}
}

\subfloat[]{~\label{fig:5 3 deterministic}
\resizebox{0.75\textwidth}{!}{
\Qcircuit @C=0.8em @R=0.8em {
    & \lstick{\ket{0}} & \qw
& \qw & \multigate{2}{G}
& \gate{X} & \multigate{2}{G^\dag}
& \qw & \qw & \multigate{2}{G}
& \gate{X} & \multigate{2}{G^\dag}\barrier[0em]{3}
& \qw \raisebox{3.5em}{$\ket{\Psi_{\text{IV}}}$}
& \qw 
& \qw
& \rstick{\ket{0}} \qw\\
    & \lstick{\ket{0}} & \gate{H} 
& \multigate{2}{\cQ_{\Uin}} & \ghost{G} 
& \multigate{2}{\cQ_{\Uin}} & \ghost{G^\dag}
& \multigate{2}{\cQ_{\Uin}} & \ctrlo{1} & \ghost{G} 
& \multigate{2}{\cQ_{\Uin}} & \ghost{G^\dag}
& \qw
& \multigate{2}{\cQ_{\Uin}} 
& \gate{H}
& \rstick{\ket{0}} \qw\\
    & \lstick{\ket{0}} & \gate{H} 
& \ghost{\cQ_{\Uin}} & \ghost{G}
& \ghost{\cQ_{\Uin}} & \ghost{G^\dag}
& \ghost{\cQ_{\Uin}} & \gate{-Z} & \ghost{G}
& \ghost{\cQ_{\Uin}} & \ghost{G^\dag}
& \qw
& \ghost{\cQ_{\Uin}} 
& \gate{H}
& \rstick{\ket{0}} \qw\\
    & \lstick{\ket{\varphi}} & \qw 
& \ghost{\cQ_{\Uin}} & \qw
& \ghost{\cQ_{\Uin}} & \qw
& \ghost{\cQ_{\Uin}} & \qw & \qw
& \ghost{\cQ_{\Uin}} & \qw
& \qw
& \ghost{\cQ_{\Uin}} 
& \qw
& \rstick{\Uin^{-1} \ket{\varphi}} \qw}
}}

\subfloat[]{
\Qcircuit @C=0.8em @R=1.8em {
    & \qw & \ctrl{2} & \qw & \ctrl{2} & \qw & \qw \\
    \lstick{Q_{\Uin} = } & \ctrl{1} & \qw & \qw & \qw & \ctrl{1} & \qw  \\
    & \gate{X} & \gate{Y} & \gate{\Uin} & \gate{Y} & \gate{X} & \qw}
}
\hspace{4em}
\subfloat[]{~\label{fig:G decomp}
\Qcircuit @C=0.8em @R=1.2em {
    & \qw & \qw & \qw & \gate{X}  & \ctrl{1} & \qw \\
    \lstick{G = } & \gate{H} & \ctrlo{1} & \gate{H} & \ctrlo{-1} & \multigate{1}{F} & \qw  \\
    & \gate{H} & \gate{-Z} & \gate{H} & \ctrlo{-1} & \ghost{F} & \qw
    }
}
\caption{ \textbf{(a)} An example of implementing $\fourcir(\Uin)$ in Theorem 1. \textbf{(b)} An example of implementing $\fivecir(\Uin)$ in Corollary 2. \textbf{(c)} Circuit implementation of $\cQ_{\Uin}$. \textbf{(d)} Circuit implementation of $G$, where $F$ is a two-qubit unitary such that $F \ket{00} = \left( \ket{01} + \ket{10} + \ket{11} \right) / \sqrt{3}$.}
\end{figure}

\setcounter{theorem}{0}
\begin{theorem}
    $\fourcir(\Uin)$ in Figure~\ref{fig:4 3 with label} satisfies
\begin{equation}
    \trace[23]{\fourcir(\Uin) \ket{000}_{123}\ket{\psi}} = \ket{0}_1 \ox \Uin^{-1} \ket{\psi}
.\end{equation}
\end{theorem}
\begin{proof}
Without loss of generality, suppose the determinant of $\Uin$ is 1 i.e., $\Uin \in \su(2)$. Notice that there is a linear relation between qubit unitary $\Uin \in \su(2)$ and its inversion
\begin{equation}
    2 \Uin^{-1} =  X \Uin X + Y \Uin Y + Z \Uin Z - \Uin
.\end{equation}
In the rest of proof, we analyze the states at four stages in Fig.~\ref{fig:4 3 with label}.
The circuit is initialized with the input state $\ket{ \Psi_{\text{in}}}$, during the first step
\begin{alignat}{1}
    \ket{ \Psi_{\text{in}}}   \xrightarrow{H} & \ket{0}  \ox \ket{+} \ox \ket{+} \ox \ket{\varphi}, \\
     \xrightarrow{\cQ_{\Uin}} & \dfrac{1}{2} \ket{0} \ox \left(\ket{00} \ox  \Uin\ket{\varphi } + \ket{01} \ox  X \Uin X \ket{\varphi } + \ket{10} \ox  Y \Uin Y\ket{\varphi } + \ket{11} \ox  Z \Uin Z\ket{\varphi } \right),  \\
     \xrightarrow{O}& \dfrac{1}{2} \ket{0} \ox \left(\ket{00} \ox  \Uin^{-1}\ket{\varphi } + \ket{01} \ox  X \Uin ^{-1}X \ket{\varphi } + \ket{10} \ox  Y \Uin ^{-1}Y\ket{\varphi } + \ket{11} \ox  Z \Uin^{-1}Z\ket{\varphi } \right).
\end{alignat}
We make the following notation to make the deduction smoother:
\begin{align}
    \ket{0^{\perp}} &:= \dfrac{\left(\ket{01} + \ket{10} + \ket{11}\right)}{\sqrt{3}} \\
    O &:= H^{\ox 2} \cdot \left( -\ketbra{0}{0} \ox Z + \ketbra{1}{1} \ox I \right)
\end{align}
It is seen that $O$ is a part of $G$ in Figure~\ref{fig:G decomp}.Applying the rest part of the defined block $G$ in Fig.~\ref{fig:G decomp} gives
\begin{alignat}{3}
\ket{\Psi_{\text{I}}} =& \dfrac{1}{2} \ket{1} \ox \ket{0^{\perp}} \ox \Uin^{-1}\ket{\varphi} + \dfrac{1}{2} \ket{0} \ox (
        &&\ket{01} \ox  X \Uin^{-1}X \ket{\varphi} + \\
{}  &{} &&\ket{10} \ox  Y \Uin^{-1}Y\ket{\varphi} + \\
{}  &{} &&\ket{11} \ox  Z \Uin^{-1} Z\ket{\varphi}), \\
 \xrightarrow{X} & \dfrac{1}{2} \ket{0} \ox \ket{0^{\perp}} \ox \Uin^{-1}\ket{\varphi} + \dfrac{1}{2} \ket{1} \ox (
        &&\ket{01} \ox  X \Uin^{-1}X \ket{\varphi} + \\
{}  &{} &&\ket{10} \ox  Y \Uin^{-1}Y\ket{\varphi} + \\
{}  &{} &&\ket{11} \ox  Z \Uin^{-1} Z\ket{\varphi}), \\
 \xrightarrow{\cQ_{\Uin}}& \dfrac{-\sqrt{3}}{2} \ket{1} \ox \ket{0^{\perp}}  \ox  \ket{\varphi}-\dfrac{1}{2\sqrt{3}} \ket{0} \ox (
        &&\ket{01} \ox  X \Uin X \Uin^{-1} \ket{\varphi} + \\
{}  &{} &&\ket{10} \ox  Y \Uin Y \Uin^{-1} \ket{\varphi}  + \\
{}  &{} &&\ket{11} \ox  Z \Uin Z \Uin^{-1} \ket{\varphi} ) .              
\end{alignat}
Applying $G^{\dag}$, it turns to
\begin{alignat}{1}
    \ket{\Psi_{\text{II}}} = \dfrac{1}{2\sqrt{3}} \ket{0} \ox [
&\ket{00} \ox (\Uin^{-1}-\Uin) + \\
&\ket{01}  \ox (\Uin + X \Uin^{-1} X)  + \\
&\ket{10}  \ox (\Uin + Y \Uin^{-1}Y)  + \\
&\ket{11}  \ox (\Uin + Z \Uin^{-1}Z) ] \Uin^{-1} \ket{\varphi}.       
\end{alignat}
Notice that $\sigma_{i} \Uin \sigma_{i}(\Uin + \sigma_{i} \Uin^{-1} \sigma_{i}) \Uin^{-1} =  \Uin^{-1} + \sigma_{i} U \sigma_{i}$, then 
\begin{alignat}{2}
    \cQ_{\Uin} \ket{\Psi_{\text{II}}} =& \dfrac{1}{2\sqrt{3}} \ket{0} \ox [
    &&\ket{00} \ox (\Uin^{-1}- \Uin) + \\
    {} & {} &&\ket{01}  \ox (\Uin^{-1} + X \Uin X)  +\\
    {} & {} &&\ket{10}  \ox (\Uin^{-1} + Y \Uin Y)  +\\
    {} & {} &&\ket{11}  \ox (\Uin^{-1} + Z \Uin Z) ]  \ket{\varphi}, \\
    \xrightarrow{c(-Z)}& \dfrac{1}{2\sqrt{3}} \ket{0} \ox [
    &&\ket{00} \ox (\Uin - \Uin^{-1}) + \\
    {} & {} &&\ket{01}  \ox (\Uin^{-1} + X \Uin X)  + \\
    {} & {} &&\ket{10}  \ox (\Uin^{-1} + Y \Uin Y)  +\\
    {} & {} &&\ket{11}  \ox (\Uin^{-1} + Z \Uin Z) ] \ket{\varphi}, \\
    \xrightarrow{O}& \dfrac{-1}{2\sqrt{3}} \ket{0} \ox [
    &&\ket{00} \ox (3 \Uin^{-1}) + \\
    {} & {} &&\ket{01}  \ox (X \Uin^{-1} X)  + \\
    {} & {} &&\ket{10}  \ox (Y \Uin^{-1} Y)  + \\
    {} & {} &&\ket{11}  \ox (Z \Uin^{-1} Z) ]\ket{\varphi}. \\         
\end{alignat}
Applying the rest of second $G$ gives
\begin{alignat}{3}
    \ket{\Psi_{\text{III}}} =& \dfrac{-\sqrt{3}}{2} \ket{0} \ox \ket{0^{\perp}} \ox \Uin^{-1}\ket{\varphi} - \dfrac{1}{2\sqrt{3}} \ket{1}  \ox (
            &&\ket{01}  \ox  X \Uin^{-1} X \ket{\varphi}  + \\
    {}  &{} &&\ket{10}  \ox Y \Uin^{-1} Y \ket{\varphi}  + \\
    {}  &{} &&\ket{11}  \ox Z \Uin^{-1} Z\ket{\varphi} ), \\
    \xrightarrow{X}& \dfrac{-\sqrt{3}}{2} \ket{1} \ox \ket{0^{\perp}} \ox \Uin^{-1}\ket{\varphi} - \dfrac{1}{2\sqrt{3}} \ket{0}  \ox (
            &&\ket{01}  \ox  X \Uin^{-1} X \ket{\varphi}  + \\
    {}  &{} &&\ket{10}  \ox Y \Uin^{-1} Y \ket{\varphi}  + \\
    {}  &{} &&\ket{11}  \ox Z \Uin^{-1} Z\ket{\varphi} ), \\
    \xrightarrow{\cQ_{\Uin}}& \dfrac{-1}{2} \ket{1} \ox \ket{0^{\perp}}  \ox  \ket{\varphi} -\dfrac{\sqrt{3}}{2} \ket{0} \ox (
    &&\ket{01} \ox  X \Uin X \Uin^{-1} \ket{\varphi} + \\
    {}  &{} &&\ket{10} \ox  Y \Uin Y \Uin^{-1} \ket{\varphi}  + \\
    {}  &{} &&\ket{11} \ox  Z \Uin Z \Uin^{-1} \ket{\varphi} ) .
\end{alignat}
Applying $G^\dag$ again. It then turns to
\begin{alignat}{3}
    \ket{\Psi_{\text{IV}}} =& \dfrac{1}{2} \ket{0} \ox (
    &&\ket{00} \ox \Uin^{-1} \Uin^{-1}\ket{\varphi} + \\
    {} & {} &&\ket{01} \ox X \Uin^{-1} X \Uin^{-1} \ket{\varphi} + \\
    {} & {} &&\ket{10} \ox Y \Uin^{-1} Y \Uin^{-1}\ket{\varphi} + \\
    {} & {} &&\ket{11} \ox Z \Uin^{-1} Z \Uin^{-1}\ket{\varphi} ), \\
    \xrightarrow{H^{\ox 2}}& \dfrac{1}{4} \ket{0} \ox [&&\ket{00} \ox 2\text{Tr}(\Uin^{-1}) \Uin^{-1} \ket{\varphi}  + \\
    {} & {} &&\ket{01} \ox  (Y \Uin^{-1} Y - \Uin) \Uin^{-1} \ket{\varphi}  + \\
    {} & {} &&\ket{10} \ox  (X \Uin^{-1} X - \Uin) \Uin^{-1} \ket{\varphi}  + \\
    {} & {} &&\ket{11} \ox  (Z \Uin^{-1} Z - \Uin) \Uin^{-1} \ket{\varphi} ], \\
    \xrightarrow{\textrm{CX} \& \textrm{CY}}& \dfrac{1}{4} \ket{0} \ox [
    &&\ket{00} \ox 2\text{Tr}(\Uin^{-1}) \Uin^{-1} \ket{\varphi} + \\
    {} & {} &&\ket{01} \ox  (\Uin^{-1}Y - Y\Uin)  \Uin^{-1} \ket{\varphi} +\\
    {} & {} &&\ket{10} \ox  (\Uin^{-1}X - X\Uin) \Uin^{-1} \ket{\varphi} + \\
    {} & {} &&\ket{11} \ox  (\Uin^{-1}Z - Z\Uin) \Uin^{-1} \ket{\varphi}]  . 
\end{alignat}
Since the decomposition of $\Uin$ on Pauli basis is $\Uin = \cos{\frac{\theta}{2}} I - i\sin{\frac{\theta}{2}} \Vec{n}\cdot\Vec{\sigma} $, it can be verified 
\begin{equation}
    \Uin^{-1} \sigma_{i} -  \sigma_{i} \Uin = 2\sin{\frac{\theta}{2}} n_{i} I
\end{equation}
Finally, the output state of the circuit is
\begin{equation}
    \ket{\Psi_{\text{OUT}}} 
    = \dfrac{1}{2} \ket{0} \ox \left(\cos{\dfrac{\theta}{2}} \ket{00}  - i\sin{\dfrac{\theta}{2}} n_{y} \ket{01} - i\sin{\dfrac{\theta}{2}} n_{x} \ket{10} - i\sin{\dfrac{\theta}{2}} n_{z} \ket{11} \right) \ox \Uin^{-1} \ket{\varphi}
\end{equation}
\end{proof}

\setcounter{corollary}{1}
\begin{corollary}
    $\fivecir$ in Figure~\ref{fig:5 3 deterministic} satisfies
    $\fivecir(U) \ket{000, \psi} = \ket{000} \ox \Uin^{-1} \ket{\psi}$.
\end{corollary}
\begin{proof}
    As shown in Figure~\ref{fig:5 3 deterministic}, the performance of the circuit $\fivecir$ on $\ket{000, \psi}$ can be directly calculated using the same method as above, where the matrices $\cQ_{\Uin}$ and $G$ are identical to those in Figure~\ref{fig:4 3 with label}. Therefore, we omit the specific calculation details here.
\end{proof}

\paragraph{\textbf{Training ansatzes}.---} In the main article, we assert that the circuit in Figure~\ref{fig:4 3 with label} is derived from the optimization results of PQComb. Figure~\ref{fig:ansatz derivation} substantiates this claim by detailing the process of circuit refinement:

\begin{figure}[H]
\centering
\subfloat[][Initial training ansatz that found 0.99 training fidelity, where each comb tooth is a 16-dimensional universal unitary tunable by 255 parameters.]{{ 
\resizebox{!}{0.012\textheight}{\Qcircuit @C=1em @R=1em {
    & \lstick{\ket{0}} & \multigate{3}{ \uni{16} } & \qw & \qw & 
\multigate{3}{ \uni{16} } & \qw & \qw_{\times 4} & \qw \backslash  \\
    & \lstick{\ket{0}} & \ghost{ \uni{16} } & \qw & \qw & 
\ghost{ \uni{16} } & \qw  & \qw & \qw \backslash  \\
    & \lstick{\ket{0}} & \ghost{ \uni{16} }  & \qw & \qw & 
\ghost{ \uni{16} } & \qw & \qw & \qw \backslash  \\
    & {} & \ghost{ \uni{16} } & \qw & \gate{\Uin} & 
\ghost{ \uni{16} } & \qw  & \qw & {}
    \gategroup{1}{5}{4}{6}{1.5em}{--}}}}
}
\hspace{2em}
\subfloat[][The ansatz further improved by guessing all entangled gates between ancilla qubits and the main system to be controlled universal qubit-gates $\textrm{U}3$.]{{ 
\resizebox{!}{0.012\textheight}{\Qcircuit @C=1em @R=1em {
    & \lstick{\ket{0}} & \multigate{2}{ \uni{8} } & \qw 
& \qw & \qw & \ctrl{3} & \qw & \ctrl{3} & \qw & \qw & \multigate{2}{ \uni{8} }
& \qw & \qw_{\times 4} & \qw & \ctrl{3} & \qw \backslash  \\
    & \lstick{\ket{0}} & \ghost{ \uni{8} } & \qw 
& \qw & \ctrl{2} & \qw & \qw & \qw & \ctrl{2} & \qw & \ghost{ \uni{8} }
& \qw & \qw & \ctrl{2} & \qw & \qw \backslash  \\
    & \lstick{\ket{0}} & \ghost{ \uni{8} }  & \qw 
& \ctrl{1} & \qw & \qw & \qw & \qw & \qw & \ctrl{1} & \ghost{ \uni{8} }
& \qw & \ctrl{1} & \qw & \qw & \qw \backslash  \\
    & {} & \qw & \qw 
& \gate{\textrm{U}3} & \gate{\textrm{U}3} & \gate{\textrm{U}3} & \gate{\Uin} & \gate{\textrm{U}3} & \gate{\textrm{U}3} & \gate{\textrm{U}3} & \qw
& \qw & \gate{\textrm{U}3} & \gate{\textrm{U}3} & \gate{\textrm{U}3} & {\qw}
    \gategroup{1}{5}{4}{12}{1.5em}{--}}}}
}
\hspace{2em}
\subfloat[][The ansatz further improved by guessing most of universal qubit-gates to be Pauli operators, and the first universal three-qubit gates to be Hadamard gates. The best training fidelity can approach 0.999 at this stage.]{{ 
\resizebox{!}{0.012\textheight}{\Qcircuit @C=1em @R=1em {
    & \lstick{\ket{0}} & \gate{H} & \qw 
& \qw & \qw & \ctrl{3} & \qw & \ctrl{3} & \qw & \qw & \multigate{2}{ \uni{8} }
& \qw & \qw_{\times 4} & \qw & \ctrl{3} & \qw \backslash  \\
    & \lstick{\ket{0}} & \gate{H} & \qw 
& \qw & \ctrl{2} & \qw & \qw & \qw & \ctrl{2} & \qw & \ghost{ \uni{8} }
& \qw & \qw & \ctrl{2} & \qw & \qw \backslash  \\
    & \lstick{\ket{0}} & \gate{H} & \qw 
& \ctrl{1} & \qw & \qw & \qw & \qw & \qw & \ctrl{1} & \ghost{ \uni{8} }
& \qw & \ctrl{1} & \qw & \qw & \qw \backslash  \\
    & {} & \qw & \qw 
& \gate{X} & \gate{Y} & \gate{Z} & \gate{\Uin} & \gate{Z} & \gate{Y} & \gate{X} & \qw
& \qw & \gate{\textrm{U}3} & \gate{\textrm{U}3} & \gate{\textrm{U}3} & {\qw}
    \gategroup{1}{5}{4}{12}{1.5em}{--}}}}
}
\hspace{2em}
\subfloat[][The ansatz further improved by guessing low coherence of the first ancilla qubit, which needs to be clean at the end: the amplitude on $\ket{0}$ is now proportional to the training fidelity.]{\label{fig:4 3 pqc} {
\resizebox{!}{0.012\textheight}{\Qcircuit @C=1em @R=1em {
    & \lstick{\ket{0}} & \qw & \qw & \qw & \qw & \qw & \qw & \qw & 
\multigate{2}{ \uni{8} } & \qw & \qw_{\times 4} & \qw & \rstick{\ket{0}} \qw  \\
    & \lstick{\ket{0}} & \gate{H} & \qw & \qw & \ctrl{2} & \qw & \ctrl{2} & \qw & 
\ghost{ \uni{8} } & \qw & \qw & \ctrl{2} & \qw \backslash  \\
    & \lstick{\ket{0}} & \gate{H}  & \qw & \ctrl{1} & \qw & \qw & \qw & \ctrl{1} & 
\ghost{ \uni{8} } & \qw & \ctrl{1} & \qw & \qw \backslash  \\
    & {} & \qw & \qw & \gate{X} & \gate{Y} & \gate{\Uin} & \gate{Y} & \gate{X} & 
\qw & \qw  & \gate{Y} & \gate{X} & \qw
    \gategroup{1}{5}{4}{10}{1.5em}{--}}}}}
\hspace{2em}
\subfloat[][The ansatz further improved by guessing the duality relations among four teeth of the comb. This ansatz can obtain fidelity close to 1, and is eventually reduced to the optimal circuit in Figure~\ref{fig:4 3 with label}.]{{
\resizebox{!}{0.012\textheight}{\Qcircuit @C=1em @R=1em {
    & \lstick{\ket{0}} & \qw
& \qw & \multigate{2}{G \coloneqq \uni{8}}\barrier[0em]{3}
& \qw &  \gate{\textrm{U}3}   & \multigate{2}{G^\dag}\barrier[0em]{3}
& \qw & \qw & \qw & \multigate{2}{G}\barrier[0em]{3}
& \qw & \gate{\textrm{U}3} & \multigate{2}{G^\dag}\barrier[0em]{3}
& \qw & \qw & \qw & \qw
& \rstick{\ket{0}} \qw\\
    & \lstick{\ket{0}} & \gate{H}
& \multigate{2}{\cQ_{\Uin}} & \ghost{G \coloneqq \uni{8}} & \qw 
& \multigate{2}{\cQ_{\Uin}}  & \ghost{G^\dag} & \qw
& \multigate{2}{\cQ_{\Uin}} & \multigate{1}{\uni{4}} & \ghost{G} & \qw
& \multigate{2}{\cQ_{\Uin}} & \ghost{G^\dag} & \qw
& \multigate{1}{\uni{4}} & \qw & \ctrl{2}
& \qw \backslash \\
    & \lstick{\ket{0}} & \gate{H} 
& \ghost{\cQ_{\Uin}} & \ghost{G \coloneqq \uni{8}} & \qw 
& \ghost{\cQ_{\Uin}}  & \ghost{G^\dag} & \qw
& \ghost{\cQ_{\Uin}} & \ghost{\uni{4}} & \ghost{G} & \qw
& \ghost{\cQ_{\Uin}} & \ghost{G^\dag} & \qw
& \ghost{\uni{4}} & \ctrl{1} & \qw
& \qw \backslash \\
    & {} & \qw 
& \ghost{\cQ_{\Uin}} & \qw & \qw
& \ghost{\cQ_{\Uin}} & \qw & \qw
& \ghost{\cQ_{\Uin}} & \qw & \qw & \qw
& \ghost{\cQ_{\Uin}} & \qw  & \qw
& \qw & \gate{Y} & \gate{X}
& \qw}}}}

\caption{Derivation of the training ansatz for the 4-call qubit-unitary inversion protocol.}~\label{fig:ansatz derivation}
\end{figure}
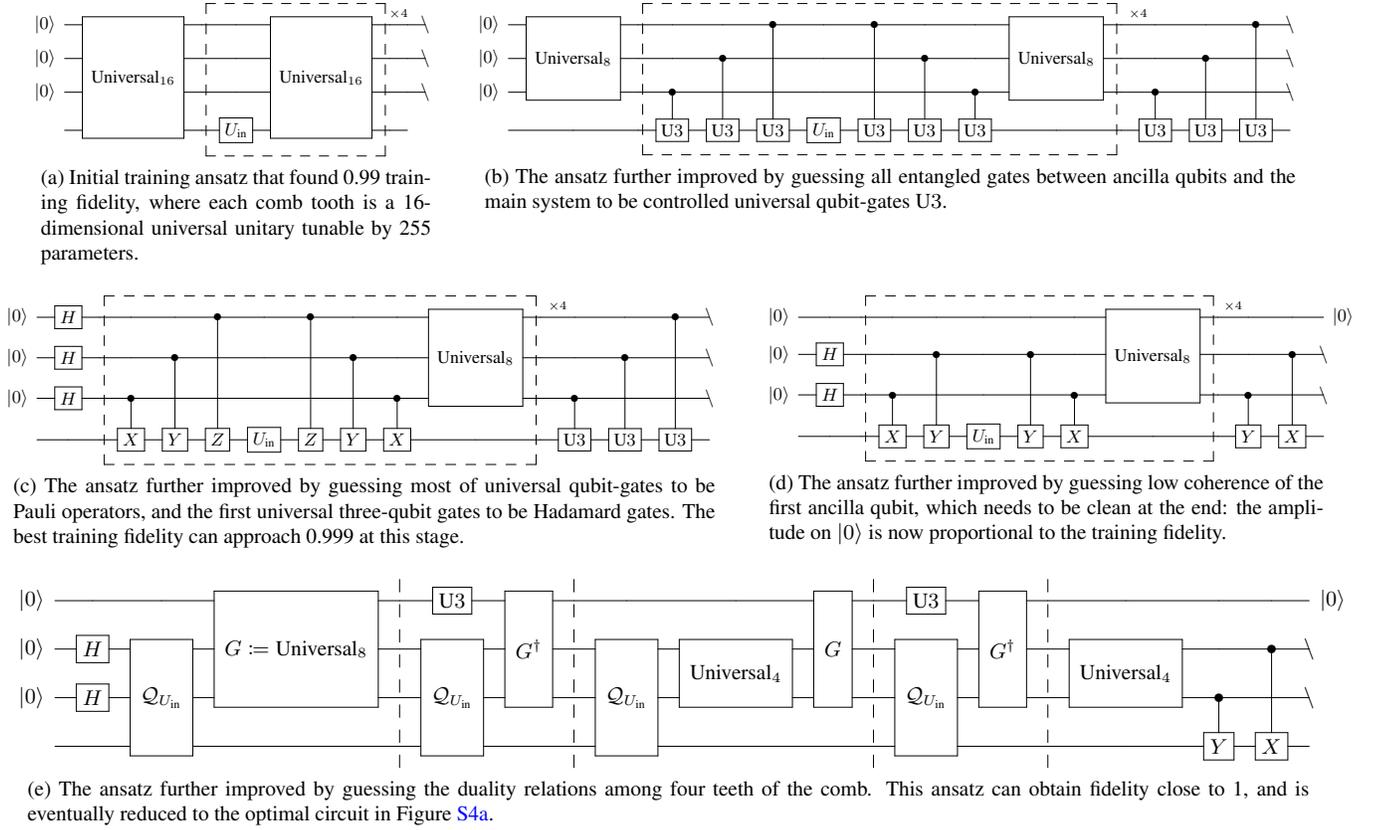

\update{

At each stage, we inspected the numerical results of the parameterized gates and got inspiration from traditional quantum algorithms. 
For example, at stage (d), we postulate that for each input $\Uin$, the surrounded universal gates $\uni{8}$ potentially include encoding and decoding operations facilitated by control Pauli operators,
enabling the ancilla qubits to produce a specific linear combination of $\Uin$ acting on the target qubit. 
This further reminded us of oblivious amplitude amplification~\cite{paetznick2014repeat}, a technique that uses the duality of oracles to amplify the amplitude of a particular state. We, therefore, conjectured that a duality relation might exist among these elements, as presented at stage (e).

\paragraph{\textbf{Protocols in noisy environments}.---} We illustrate how the IBM simulation in Figure 3 is completed. We require that our experiment should work for any input state. However, since all operation in the circuit may contain error, it is not appropriate to compute the average similarities using the comb formalism. As a compromised solution, one can add an extra qubit to the main system, and the input state of which is taken as the Bell state $\ket{\Psi} = \left(\ket{00} + \ket{11}\right) / \sqrt{2}$. The similarity for each input $\Uin$ is the state fidelity between the output state and $\Uin^{-1} \ket{\Psi}$. Then the average similarity is approximated by sampling 200 unitaries.

We additionally consider the scenario where the input unitary gates are imperfect, being affected by depolarizing noise characterized by $D_p(\rho) = p I / 2 + (1 - p) \rho$ with the noise level $p$. 
Consequently, the circuit receives $D_p \circ \cU_\textrm{in}$ as input, and the output deviates from the intended $\Uin^{-1}$. 
To evaluate the performance under this noise level, we adopted the same approach as in calculating the comb-based loss function. Based on the specific circuits for the 4-call protocol and the 5-call protocol, we first get their Choi operators. 
We then calculate the performance operator under noise level $p$, denoted as $\O_p$, by substituting the input processes with $D_p \circ \cU_\textrm{in}$,
\begin{equation}
    \O^{\textrm{DE}}_p = \integral[\textrm{Haar}]{U} \ketbravec{U^{-1}}{U^{-1}}_{\bmP, \bmF} \ox \left(\cJ_{\cU \circ D_p}\right)_{I_1, O_1}^T \ox \ldots \ox \left(\cJ_{\cU \circ D_p}\right)_{I_m, O_m}^T
,\end{equation}
where $\cU: \rho \to U \rho U^\dag$ is the operational representation of $U$ and $\cJ_{\cU \circ D_p}$ is the Choi operator of the composition between operation $\cU$ and the depolarzing noise $D_p$. In the experiment, $\O^{\textrm{DE}}_p$ is obtained by sampling $40000$ unitaries. For every quantum comb that inputs such noisy unitary, the corresponding average similarity function between its output and the inverse unitary can be evaluated by $\trace{\O^{\textrm{DE}}_p C_{\bmV}}$, where $C_{\bmV}$ is the Choi operator of this comb. 

The results, illustrated in Figure~\ref{fig:depo noise}, show a near-linear decline in average similarity for both circuits as the noise level $p$ increases. Notably, the $4$-call protocol exhibits better performance in comparison to the $5$-call protocol, maintaining an average similarity above $0.9$ for $p < 0.05$.

\begin{figure}[H]
    \centering
    \includegraphics[width=0.4\linewidth]{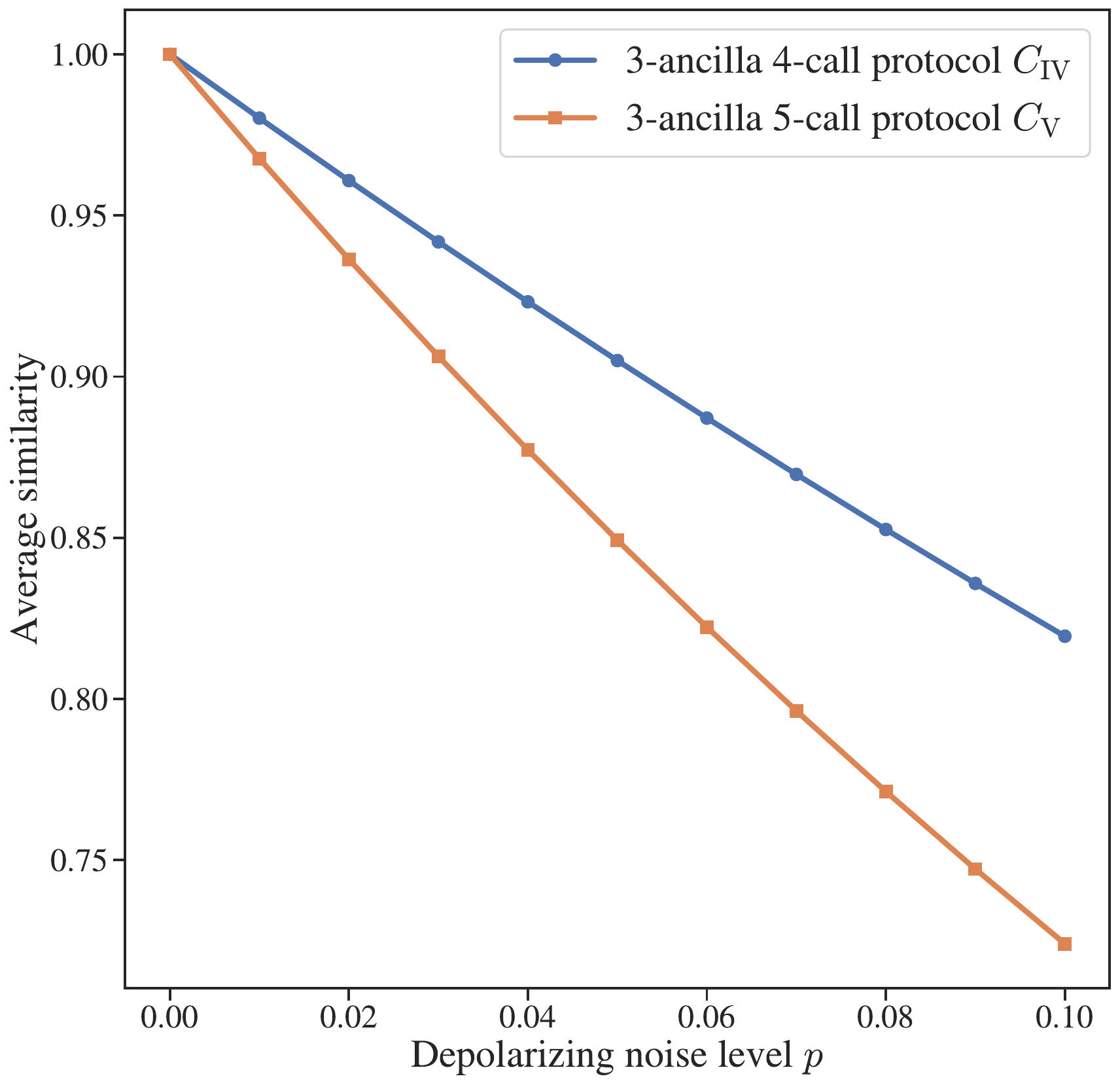}
    \caption{Simulation of our two protocols under the noise setting that input unitary is carried by a depolarizing noise with noise level $p \in [0, 0.1]$.}
    \label{fig:depo noise}
\end{figure}
}

\end{document}

%% file: pretex.tex
\usepackage{amsfonts,amsmath,amssymb,mathtools,mathrsfs,bbm,float}
\usepackage[ruled, vlined, linesnumbered]{algorithm2e}
\usepackage{graphicx,epic,eepic,epsfig,latexsym,verbatim,color}
\usepackage[inline,shortlabels]{enumitem}
\usepackage[caption=false]{subfig}
\usepackage[dvipsnames]{xcolor}
 
\usepackage{tikz}
\usetikzlibrary{chains}
\usetikzlibrary{fit}

\usepackage{epsfig}
\usetikzlibrary{shapes.symbols,patterns} % for source symbols
\usepackage{pgfplots}

\usepackage[strict]{changepage}

\usepackage[marginal]{footmisc}
\usepackage{url}
\usepackage{theorem}
\usepackage{thmtools}
\usepackage[ISO]{diffcoeff}
\usepackage{nicematrix}
\usepackage{xfrac}

\newtheorem{proposition}{Proposition}

\newtheorem{theorem}[proposition]{Theorem}
\newtheorem{corollary}[proposition]{Corollary}

\newenvironment{proof}{\noindent \textbf{{Proof~} }}{\hfill $\blacksquare$}
\def\squareforqed{\hbox{\rlap{$\sqcap$}$\sqcup$}}
\def\qed{\ifmmode\squareforqed\else{\unskip\nobreak\hfil
\penalty50\hskip1em\null\nobreak\hfil\squareforqed
\parfillskip=0pt\finalhyphendemerits=0\endgraf}\fi}
\def\endenv{\ifmmode\;\else{\unskip\nobreak\hfil
\penalty50\hskip1em\null\nobreak\hfil\;
\parfillskip=0pt\finalhyphendemerits=0\endgraf}\fi}

\newcounter{example}

\mathchardef\ordinarycolon\mathcode`\:
\mathcode`\:=\string"8000
\def\vcentcolon{\mathrel{\mathop\ordinarycolon}}
\begingroup \catcode`\:=\active
  \lowercase{\endgroup
  \let :\vcentcolon
  }

\usepackage{hyperref}
\hypersetup{colorlinks=true,citecolor=blue,linkcolor=blue,filecolor=blue,urlcolor=blue,breaklinks=true}
\usepackage{footnotebackref}
\usepackage{cleveref}
\usepackage{graphicx}

\definecolor{darkblue}{RGB}{0,76,156}
\definecolor{darkkblue}{RGB}{0,0,153}
\definecolor{blue2}{RGB}{102,178,255}
\definecolor{darkred}{RGB}{195,0,0}

\RequirePackage[framemethod=default]{mdframed}
\newmdenv[skipabove=7pt,
skipbelow=7pt,
backgroundcolor=darkblue!15,
innerleftmargin=5pt,
innerrightmargin=5pt,
innertopmargin=5pt,
leftmargin=0cm,
rightmargin=0cm,
innerbottommargin=5pt,
linewidth=1pt]{tBox}

\newmdenv[skipabove=7pt,
skipbelow=7pt,
backgroundcolor=blue2!25,
innerleftmargin=5pt,
innerrightmargin=5pt,
innertopmargin=5pt,
leftmargin=0cm,
rightmargin=0cm,
innerbottommargin=5pt,
linewidth=1pt]{dBox}

\newmdenv[skipabove=7pt,
skipbelow=7pt,
backgroundcolor=darkred!15,
innerleftmargin=5pt,
innerrightmargin=5pt,
innertopmargin=5pt,
leftmargin=0cm,
rightmargin=0cm,
innerbottommargin=5pt,
linewidth=1pt]{rBox}

\newcommand{\nc}{\newcommand}
\nc{\ketbra}[2]{\lvert#1\rangle\!\langle#2\rvert}
\nc{\braket}[2]{\langle#1|#2\rangle}

\DeclarePairedDelimiter{\norm}{\lVert}{\rVert}
\DeclarePairedDelimiter{\abs}{\lvert}{\rvert}

\makeatletter
\let\oldabs\abs
\def\abs{\@ifstar{\oldabs}{\oldabs*}}
\let\oldnorm\norm
\def\norm{\@ifstar{\oldnorm}{\oldnorm*}}
\makeatother

\DeclarePairedDelimiterX{\infdivx}[2]{(}{)}{%
  #1\;\delimsize\|\;#2%
}

\nc{\proj}[1]{| #1\rangle\!\langle #1 |}
\nc{\avg}[1]{\langle#1\rangle}
\nc{\smfrac}[2]{\mbox{$\frac{#1}{#2}$}}
\nc{\tr}{\operatorname{Tr}}
\nc{\ox}{\otimes}
\nc{\dg}{\dagger}
\nc{\dn}{\downarrow}
\nc{\cA}{{\cal A}}
\nc{\cB}{{\cal B}}
\nc{\cC}{{\cal C}}
\nc{\cD}{{\cal D}}
\nc{\cE}{{\cal E}}
\nc{\cF}{{\cal F}}
\nc{\cG}{{\cal G}}
\nc{\cH}{{\cal H}}
\nc{\cI}{{\cal I}}
\nc{\cJ}{{\cal J}}
\nc{\cK}{{\cal K}}
\nc{\cL}{{\cal L}}
\nc{\cM}{{\cal M}}
\nc{\cN}{{\cal N}}
\nc{\cO}{{\cal O}}
\nc{\cP}{{\cal P}}
\nc{\cQ}{{\cal Q}}
\nc{\cR}{{\cal R}}
\nc{\cS}{{\cal S}}
\nc{\cT}{{\cal T}}
\nc{\cU}{{\cal U}}
\nc{\cV}{{\cal V}}
\nc{\cX}{{\cal X}}
\nc{\cY}{{\cal Y}}
\nc{\cZ}{{\cal Z}}
\nc{\cW}{{\cal W}}
\nc{\csupp}{{\operatorname{csupp}}}
\nc{\qsupp}{{\operatorname{qsupp}}}
\nc{\var}{{\operatorname{var}}}
\nc{\rar}{\rightarrow}
\nc{\lrar}{\longrightarrow}
\nc{\polylog}{{\operatorname{polylog}}}
\nc{\wt}{{\operatorname{wt}}}
\nc{\supp}{{\operatorname{supp}}}

\nc{\argmin}{{\operatorname{argmin}}}

\makeatletter
\newcommand{\tpmod}[1]{{\@displayfalse\pmod{#1}}}
\makeatother

\def\x{\xi}

\def\O{\Omega}

\nc{\RR}{{{\mathbb R}}}
\nc{\CC}{{{\mathbb C}}}
\nc{\FF}{{{\mathbb F}}}
\nc{\NN}{{{\mathbb N}}}
\nc{\ZZ}{{{\mathbb Z}}}
\nc{\PP}{{{\mathbb P}}}
\nc{\QQ}{{{\mathbb Q}}}
\nc{\UU}{{{\mathbb U}}}
\nc{\EE}{{{\mathbb E}}}
\nc{\id}{{\operatorname{id}}}

\nc{\CHSH}{{\operatorname{CHSH}}}

\nc{\rU}{\mbox{U}}

\nc{\ob}[1]{#1}

\nc{\SEP}{{\text{\rm SEP}}}
\nc{\NS}{{\text{\rm NS}}}
\nc{\LOCC}{{\text{\rm LOCC}}}
\nc{\PPT}{{\text{\rm PPT}}}
\nc{\EXT}{{\text{\rm EXT}}}
\nc{\Sym}{{\operatorname{Sym}}}

\nc{\ERLO}{{E_{\text{r,LO}}}}
\nc{\ERLOCC}{{E_{\text{r,LOCC}}}}
\nc{\ERPPT}{{E_{\text{r,PPT}}}}
\nc{\ERLOCCinfty}{{E^{\infty}_{\text{r,LOCC}}}}
\nc{\Aram}{{\operatorname{\sf A}}}

\newcommand{\Choi}{Choi-Jamio\l{}kowski }

\makeatletter
\def\grd@save@target#1{%
  \def\grd@target{#1}}
\def\grd@save@start#1{%
  \def\grd@start{#1}}
\tikzset{
  grid with coordinates/.style={
    to path={%
      \pgfextra{%
        \edef\grd@@target{(\tikztotarget)}%
        \tikz@scan@one@point\grd@save@target\grd@@target\relax
        \edef\grd@@start{(\tikztostart)}%
        \tikz@scan@one@point\grd@save@start\grd@@start\relax
        \draw[minor help lines,magenta] (\tikztostart) grid (\tikztotarget);
        \draw[major help lines] (\tikztostart) grid (\tikztotarget);
        \grd@start
        \pgfmathsetmacro{\grd@xa}{\the\pgf@x/1cm}
        \pgfmathsetmacro{\grd@ya}{\the\pgf@y/1cm}
        \grd@target
        \pgfmathsetmacro{\grd@xb}{\the\pgf@x/1cm}
        \pgfmathsetmacro{\grd@yb}{\the\pgf@y/1cm}
        \pgfmathsetmacro{\grd@xc}{\grd@xa + \pgfkeysvalueof{/tikz/grid with coordinates/major step}}
        \pgfmathsetmacro{\grd@yc}{\grd@ya + \pgfkeysvalueof{/tikz/grid with coordinates/major step}}
        \foreach \x in {\grd@xa,\grd@xc,...,\grd@xb}
        \node[anchor=north] at (\x,\grd@ya) {\pgfmathprintnumber{\x}};
        \foreach \y in {\grd@ya,\grd@yc,...,\grd@yb}
        \node[anchor=east] at (\grd@xa,\y) {\pgfmathprintnumber{\y}};
      }
    }
  },
  minor help lines/.style={
    help lines,
    step=\pgfkeysvalueof{/tikz/grid with coordinates/minor step}
  },
  major help lines/.style={
    help lines,
    line width=\pgfkeysvalueof{/tikz/grid with coordinates/major line width},
    step=\pgfkeysvalueof{/tikz/grid with coordinates/major step}
  },
  grid with coordinates/.cd,
  minor step/.initial=.2,
  major step/.initial=1,
  major line width/.initial=2pt,
}
\makeatother

\usepackage{thmtools}
\usepackage{thm-restate}
\usepackage{etoolbox}
\makeatletter
\def\problem@s{}
\newcounter{problems@cnt}

\newcommand{\allproblems}{\problem@s}
\makeatother